\documentclass[11pt, letterpaper]{article}
\pdfoutput=1
\usepackage{fullpage}
\usepackage[OT4]{fontenc}
\usepackage{amsmath,amsfonts,amsthm}
\usepackage[figure,boxed,lined]{algorithm2e}
\usepackage{color}
\usepackage{enumerate}
 \usepackage{graphicx}
 
\usepackage{hyperref}
\newtheorem{theorem}{Theorem}[section]
\newtheorem{corollary}[theorem]{Corollary}
\newtheorem{lemma}[theorem]{Lemma}

\newtheorem{definition}[theorem]{Definition}

\newtheorem{fact}[theorem]{Fact}

\newcommand{\chist}{\chi_{s,t}}
\newcommand{\bchist}{\boldsymbol{\mathit{\chi_{s,t}}}}

\newcommand{\maxflow}{F^{*}}
\newcommand{\eps}{\epsilon}

\newcommand{\reff}{\mathrm{Reff}_{st}}
\newcommand{\Reff}[1]{R_{\mathrm{eff}}(#1)}

\newcommand{\Ceff}[1]{C_{\mathrm{eff}}(#1)}
\newcommand{\Ot}[1]{\widetilde{O}\left(#1 \right)}

\def\defeq{\stackrel{\mathrm{def}}{=}}

\newdimen\pIR
\newcommand{\StevesR}{{\rm I\kern\pIR R}}
\def\Reals#1{\StevesR^{#1}}

\def\norm#1{\left\| #1 \right\|}

\def\abs#1{\left|#1  \right|}

\def\pinv#1{{#1}^{\dagger}}

\newcommand{\ff}{\boldsymbol{\mathit{f}}}

\newcommand{\ii}{\boldsymbol{\mathit{i}}}

\newcommand{\rr}{\boldsymbol{\mathit{r}}}

\newcommand{\ww}{\boldsymbol{\mathit{w}}}

\newcommand{\BB}{\boldsymbol{\mathit{B}}}
\newcommand{\CC}{\boldsymbol{\mathit{C}}}

\newcommand{\LL}{\boldsymbol{\mathit{L}}}

\newcommand{\RR}{\boldsymbol{\mathit{R}}}

\newcommand{\thr}[1]{|#1|}
\newcommand{\energy}[2]{\mathcal{E}_{#1}(#2)}
\newcommand{\phiphi}{\boldsymbol{\mathit{\phi}}}
\newcommand{\tphi}{\boldsymbol{\widetilde{\mathit{\phi}}}}

\newcommand{\tf}{\widetilde{f}}
\newcommand{\hf}{\widehat{f}}
\newcommand{\hff}{\boldsymbol{\mathit{\hf}}}
\newcommand{\tff}{\boldsymbol{\mathit{\tf}}}
\newcommand{\of}{\overline{f}}
\newcommand{\oracle}{\mathit{O}}
\renewcommand{\cong}[2]{\mathsf{cong}_{#1}(#2)}

\newcommand{\hphiphi}{\boldsymbol{\widehat{\phi}}}
\newcommand{\tphiphi}{\boldsymbol{\tilde{\phi}}}
\def\maxflowguess{F}

\newcommand{\rhop}{\rho}

\newcommand{\rrt}[1]{\rr^{#1}}
\newcommand{\rt}[2]{r^{#1}_{#2}}
\newcommand{\flowt}[1]{f_{\rr^{#1}}}
\newcommand{\newcode}[1]{\colorbox{yellow}{#1}}
\newcommand{\fail}{ \textbf{``fail''}}

\begin{document}
\title{Electrical Flows, Laplacian Systems, and Faster Approximation of Maximum Flow in Undirected Graphs}

\author{ Paul Christiano \\ MIT \and 
Jonathan A.\ Kelner\thanks{Research partially supported by NSF grant CCF-0843915.}\\ MIT \and 
Aleksander M\k{a}dry\thanks{Research partially supported by NSF grant CCF-0829878 and by ONR grant N00014-05-1-0148.} \\ MIT \and 
Daniel Spielman\thanks{This material is based upon work supported by the National Science Foundation 
  under Grant Nos. 0634957, 0634904 and 0915487.
Any opinions, findings, and conclusions or recommendations expressed in this material are those of the authors and do not necessarily reflect the views of the National Science Foundation.
}%
 \\ Yale University \and 
Shang-Hua Teng\thanks{This research is in part supported by NSF
  grants 1032367, 0964481, and a USC Viterbi School of Engineering
  startup grant which is in turn
  supported by a Powell Foundation Award. } \\ University of Southern California
}

\maketitle

\date{}
\thispagestyle{empty}
\begin{abstract}

We introduce a new approach to computing an approximately 
  maximum $s$-$t$ flow in a capacitated, undirected graph.
This flow is computed by solving a sequence of electrical flow problems.
Each electrical flow is 
   given by
  the solution of a system of linear equations
  in a Laplacian matrix, and thus may be 
  approximately computed in nearly-linear time.

Using this approach, we develop the fastest known
  algorithm for computing approximately maximum $s$-$t$ flows.
For a graph having $n$ vertices and $m$ edges,
  our algorithm computes a $(1-\epsilon)$-approximately maximum $s$-$t$ flow
  in time\footnote{We recall that $\Ot{f (m)}$ denotes $O (f (m) \log^{c} f (m))$ for some constant $c$.}  $\Ot{mn^{1/3} \epsilon^{-11/3}}$.
A dual version of our approach computes 
    a $(1+\epsilon)$-approximately minimum
  $s$-$t$ cut in time
   $\Ot{m+n^{4/3}\eps^{-8/3}}$, which is the fastest known algorithm
  for this problem as well.
Previously, the best dependence on $m$ and $n$ was
  achieved by the algorithm of Goldberg and Rao (J. ACM 1998),
  which can be used to compute approximately 
  maximum  $s$-$t$ flows in time $\Ot{m\sqrt{n}\epsilon^{-1}}$,\
  and
  approximately minimum $s$-$t$ cuts in time $\Ot{m+n^{3/2}\epsilon^{-3}}$.

\end{abstract}


\newpage
\setcounter{page}{1}
\section{Introduction} 

The maximum $s$-$t$ flow problem and its dual, the minimum $s$-$t$ cut problem,
  are two of the most fundamental and extensively
  studied problems in Operations Research and Optimization~\cite{SchrijverA,AhujaBook}.
They have many applications (see~\cite{AhujaApps})
  and are often used as subroutines in other algorithms (see~\cite{AroraHK05,ShermanBreaking}).
Many advances have been made in the development of algorithms for this
  problem (see Goldberg and Rao \cite{GoldbergRao} for an overview).
However, for the basic problem of computing or $(1-\epsilon)$-approximating the maximum flow
  in undirected, unit-capacity  graphs with $m = O (n)$ edges,
  the asymptotically fastest known algorithm is the one developed
  in 1975 by 
  Even and Tarjan \cite{EvenTarjan}, which takes time $O(n^{3/2})$.
Despite 35 years of extensive work on the problem, this bound has not been improved.

In this paper, we introduce a new approach to computing  approximately maximum
  $s$-$t$ flows and minimum $s$-$t$ cuts in undirected, capacitated graphs.  
Using it, we present the first algorithms that break the $O(n^{3/2})$ 
  complexity barrier described above.  
In addition to being the fastest known algorithms for this problem, they are simple
  to describe and introduce techniques that may be applicable to other problems.
In them, we reduce the problem of computing maximum flows subject to
  capacity constraints to the problem of computing electrical flows
  in resistor networks.
An approximate solution to each  electrical flow problem can be found
  in time $\Ot{m}$ using recently developed
  algorithms for solving systems of linear equations in Laplacian 
  matrices~\cite{KoutisMP10,SpielmanTeng_solving}.

There is a simple physical intuition that underlies our approach,
which we describe here in the case of a graph with unit edge capacities.
We begin by thinking of each edge of the input graph as a resistor with resistance one,
 and we 
  compute the electrical flow
  that results when we send current from the source 
  $s$ to the sink $t$.
These currents obey the flow conservation constraints, but they may not respect the
  capacities of the edges.
To remedy this, we increase the resistance of each edge in proportion to
  the amount of current flowing through it---thereby
 penalizing edges that violate their capacities---and
 compute the electrical flow with these new resistances.
 
After repeating this operation $\Ot{m^{1/3}\cdot \mathrm{poly(1/\epsilon)}}$ times,
we will be able to 
obtain a $(1-\epsilon)$-approximately maximum $s$-$t$ flow by taking a certain average
of the electrical flows that we have computed,
and we will be able to extract a $(1+\epsilon)$-approximately minimum $s$-$t$ cut from the vertex potentials\footnote{For clarity, we will analyze these two cases separately, and they will use slightly different rules for updating the resistances.}.
This will give us algorithms for both problems that run in time $\Ot{m^{4/3}\cdot \mathrm{poly}(1/\epsilon)}$. By combining this with the graph smoothing and sampling techniques of Karger \cite{Karger98}, we can get a $(1-\epsilon)$-approximately maximum $s$-$t$ flow in time $\Ot{mn^{1/3}\epsilon^{-11/3}}$.
Furthermore, by applying the cut algorithm to a sparsifier 
  \cite{BenczurK96} of the input graph,
  we can compute a $(1+\epsilon)$-approximately minimum $s$-$t$ 
  cut in time $\Ot{m + n^{4/3} \epsilon^{-8/3}}$.

We remark that the results in this paper immediately improve the running time of algorithms
  that use the computation of an approximately maximum $s$-$t$ flow on an undirected,
  capacitated graph as a subroutine.  For example, combining our work with that of
  Sherman~\cite{ShermanBreaking}
  allows us to
  achieve the best currently known approximation ratio of $O(\sqrt{\log{n}})$ for the
  sparsest cut problem in time $\Ot{m+n^{4/3}}$.

We are hopeful that our approach can be  extended to directed graphs and
  can also eventually lead to an algorithm
   that approximately solves the maximum flow problem
   in nearly-linear time.

\subsection{Previous Work on Maximum Flows and Minimum Cuts} 

The best previously known algorithms for the problems studied here are
due to Goldberg and Rao. 
In a breakthrough paper, Goldberg and Rao~\cite{GoldbergRao}
  developed an algorithm for computing exact maximum $s$-$t$ flows
  in directed or undirected capacitated graphs in time
$
O (m \min (n^{2/3}, m^{1/2}) \log (n^{2}/m) \log U),
$
  assuming that the edge capacities are 
  integers between $1$ and $U$.
When we are interested in finding $(1-\epsilon)$-approximately maximum $s$-$t$ flow, the dependence on $\log U$ can be removed and, by employing the smoothing technique of Karger \cite{Karger98}, one can obtain a 
 running time of
   $$\Ot{m\sqrt{n}\epsilon^{-1}}.$$
By applying their algorithm to a sparsifier, as constructed by 
Bencz\'{u}r and Karger \cite{BenczurK96}, 
Goldberg and Rao show how to compute a $(1+\epsilon)$-approximately
  minimum $s$-$t$ cut in an undirected graph in time
   $$\Ot{   m +   n^{3/2} \epsilon^{-3}}.$$
Their work was the culmination of a long line of papers on the problem;
we refer the reader to their paper for an extensive survey of 
  earlier developments in algorithms for computing maximum $s$-$t$ flows.
 In more recent work, Daitch and Spielman~\cite{DaitchSpielman} showed that
  fast solvers for Laplacian linear systems~\cite{SpielmanTeng_solving,KoutisMP10}
  could be used to make
  interior-point algorithms for the maximum flow and minimum-cost flow
  problems run in time $\Ot{m^{3/2} \log U}$, 
  and M\k{a}dry~\cite{Madry10} showed that one can approximate a wide range of cut problems,
  including the minimum $s$-$t$ cut problem, within a polylogarithmic factor in almost linear
  time.

\subsection{Outline}
We begin the technical part of this paper
  in Section~\ref{sec:ElectricalFlow} with a review of 
  maximum flows and electrical flows, along with several
  theorems about them that we will need in the sequel. 
In Section~\ref{sec:simple} we give a simplified version of our
  approximate maximum-flow
  algorithm that has running time $\Ot{m^{3/2} \epsilon^{-5/2}}$.
   In Section~\ref{sec:improved_flow}, we will show how to improve the running time of our algorithm to  $\Ot{m^{4/3} \epsilon^{-3}}$; we will then describe how to combine this with existing graph smoothing and sparsification techniques to   
compute approximately maximum $s$-$t$ flows  in time $\Ot{mn^{1/3}\epsilon^{-11/3}}$ and to approximate the value of such flows in time  $\Ot{m+n^{4/3} \epsilon^{-8/3}}$.
In Section~\ref{sec:dual_alg}, we present
  a variant of our algorithm that computes approximately minimum $s$-$t$
  cuts in time $\Ot{m+n^{4/3} \epsilon^{-8/3}}$.


\section{Maximum Flows, Electrical Flows, and Laplacian Systems}\label{sec:ElectricalFlow}

\subsection{Graph Theory Definitions}

Throughout the rest of the paper, let $G=(V,E)$ be an undirected graph with $n$ vertices and $m$ edges.   We distinguish two vertices, a \emph{source} vertex $s$ and a \emph{sink} vertex $t$.  
We assign each edge $e$ a nonzero integral \emph{capacity} $u_{e} \in  \mathbb{Z}^{+}$, and we let $U:=\max_e
  {u_e}/\min_e{u_e}$ be the ratio of the largest to the smallest capacities.  

We arbitrarily orient each edge in $E$; this divides the edges incident to a  
  vertex $v\in V$ into the set $E^{+}(v)$ of edges oriented towards
  $v$ 
   and the set $E^{-}(v)$ of edges 
  oriented away from $v$. These orientations are merely for notational convenience.
We use them to interpret the meaning of a positive flow
  on an edge.
If an edge has positive flow and is in $E^{+} (v)$, then the flow is
  towards $v$.
Conversely, if it has negative flow then the flow is away from $v$.
One should keep in mind that our graphs are undirected and that the flow on an
  edge can go in either direction, regardless of this edge's orientation.

We now define our primary objects of study, $s$-$t$ cuts and $s$-$t$ flows. 
\begin{definition}[Cuts]
  An \emph{$s$-$t$ cut} is a partition $(S,V \setminus S)$ of the vertices into two disjoint sets such that $s\in  S $
     and $t\in V\setminus S$.  
  The \emph{capacity $u(S)$} of the cut is defined to be the sum 
    $u(S):=\sum_{e\in E(S)} u_e,$
    where $E(S)\subseteq E$ is the set of edges with one endpoint in $S$ and one endpoint in $V\setminus S$.
\end{definition}
  
\begin{definition}[Flows]
  An \emph{$s$-$t$ flow} is a function $f:E\rightarrow \Reals{}$ that obeys the \emph{flow-conservation constraints}
  \[
     \sum_{e\in E^{-}(v)} f(e)-\sum_{e\in E^{+}(v)} f(e)=0 \quad\text{for all $v\in V\setminus\{s,t\}$}.
  \]
  The \emph{value $|f|$} of the flow is defined to be the net flow out of the source vertex, 
  $
    |f|:=\sum_{e\in E^{-}(s)} f(e)-\sum_{e\in E^{+}(s)} f(e).
  $
\end{definition}
It follows easily from the flow conservation constraints that the net flow out of $s$ is equal to the net flow into $t$, so $|f|$ may be  
  interpreted as the amount of flow that is sent from $s$ to $t$. 

\subsection{Maximum Flows and Minimum Cuts}\label{sec:max_flow_prelim}

We say that an $s$-$t$ flow $f$ is \emph{feasible} if $|f(e)|\leq u_e$ for each edge $e$, i.e., if the amount of flow routed through any edge does not exceed its capacity.   
The \emph{maximum $s$-$t$ flow problem} is that of finding a feasible $s$-$t$ flow in $G$ of maximum value.  
We denote a maximum flow in $G$ (with the given capacities) by $f^*$, 
  and we denote its value by $F^*:=|f^*|$.
We say that $f$ is a $(1-\epsilon)$-approximately maximum flow if it is a feasible $s$-$t$ flow
  of value at least $(1-\epsilon) F^{*}$.

To simplify the exposition, we will take $\epsilon$ to be a constant independent of $m$ throughout the paper, and $m$ will be assumed to be larger than some fixed constant.  However, our analysis will go through unchanged as long as $\epsilon>\widetilde{\Omega}{(m^{-1/3}})$.  In particular, our analysis will apply to all $\epsilon$ for which our given bounds are faster than the $O(m^{3/2})$ time required by existing exact algorithms.

One can easily reduce the problem of finding a $(1-\epsilon)$-approximation to the  maximum
  flow in an arbitrary undirected graph to that of finding a $(1-\epsilon/2)$-approximation
  in a graph in which the ratio of the largest to smallest capacities is polynomially 
  bounded.
To do this, one should first compute a crude approximation of the maximum flow in
  the original graph.
For example, one can compute the $s$-$t$ path of maximum bottleneck in
  time $O (m + n \log n)$~\cite[Section 8.6e]{SchrijverA}, 
  where we recall that the bottleneck of a path is the 
  minimum capacity of an edge on that path.
If this maximum bottleneck of an $s$-$t$ path is $B$, then the maximum flow
  lies between $B$ and $m B$.
This means that there is a maximum flow in which each edge flows
  at most $m B$, so all capacities can be decreased to be at most $m B$.
On the other hand, if one removes all the edges with capacities less $\epsilon B /2 m$,
  the maximum flow can decrease by at most $\epsilon B /2$.
So, we can assume that the minimum capacity is at least $\epsilon B / 2m$
  and the maximum is at most $B m$, for a ratio of $2 m^{2} / \epsilon$.
Thus, by a simple scaling, we can assume that all edge capacities
  are integers between $1$ and $2 m^{2} / \epsilon$.

The \emph{minimum $s$-$t$ cut problem} is that of finding the $s$-$t$ cut of minimum capacity.  
The \emph{Max Flow-Min Cut Theorem} (\cite{FordF56,EliasFS56}) states that the capacity of the minimum $s$-$t$ cut is equal to $F^*$, the value of the 
  maximum $s$-$t$ flow.  

In particular, the Max Flow-Min Cut Theorem implies that one can use the capacity of any $s$-$t$ cut as an upper bound on the value 
  of any feasible $s$-$t$ flow, 
   and that the task of finding the \emph{value} of the maximum flow is equivalent to the task of finding the  
  \emph{capacity} of a minimal $s$-$t$ cut. 
  
One should note, however, that the above equivalence applies only to the values of the flow and the capacity and that although one can easily obtain a minimum $s$-$t$ cut of a graph given its maximum flow, there is no known procedure that obtains a maximum flow from minimum $s$-$t$ cut more efficiently than by just computing the maximum flow from a scratch.

\subsection{Electrical Flows and the Nearly Linear Time Laplacian Solver}
In this section, we review some basic facts about electrical flows in networks of resistors and present
  a theorem that allows us to quickly approximate these flows.
For an in-depth treatment of the background material, we refer the reader to~\cite{Bollobas98}.

We begin by assigning a \emph{resistance $r_e>0$} to each edge $e\in E$, and we collect these resistances into a vector $\rr\in \Reals{m}$.   
For a given $s$-$t$ flow $f$, we define its \emph{energy} (with respect to $\rr$) as 
$$\energy{\rr}{f}:=\sum_{e} r_{e} f^2(e).$$ 
The \emph{electrical flow of value $F$ (with respect to $\rr$, from $s$ to $t$)}
is the flow that minimizes $\energy{\rr}{f}$ among all $s$-$t$ flows $f$ of value $F$. 
This flow is easily shown to be unique, and we note that it need not respect the capacity constraints.

From a physical point of view, the electrical flow of value one corresponds to the current that is induced in $G$ if we view it as an electrical circuit in which each edge $e$ has resistance of $r_{e}$, and we send one unit of current from $s$ to $t$, say by attaching $s$ to a current source and $t$ to ground.  

\subsubsection{Electrical Flows and Linear Systems}\label{sec:electrical_flows}

While finding the maximum $s$-$t$ flow corresponds to solving a linear program, we can compute the electrical flow by solving a system 
  of linear equations.  
To do so, we introduce the \emph{edge-vertex incidence matrix} $\BB$, which is an $n\times m$ matrix with rows indexed by vertices 
  and columns indexed by edges, such that
\[
  \BB_{v,e} = \begin{cases}
1  & \text{if $e\in E^{-}(v)$,}\\
-1  & \text{if $e\in E^{+}(v)$,}\\
0 & \text{otherwise.}
\end{cases}
\]

If we treat our flow $f$ as a vector $\ff\in\Reals{m}$, 
  where we use the orientations of the edges to determine the signs of the coordinates, 
  the $v^{\text{th}}$ entry of the vector $\BB^T \ff$ will be the difference between the flow out of and the flow into vertex $v$.  
As such, the constraints that one unit of flow is sent from $s$ to $t$ and that flow is conserved at all other vertices can be written as
  $$ B^T \ff =\bchist,$$
  where $\bchist$ is the vector with a $1$ in the coordinate corresponding to $s$, a $-1$ in the coordinate corresponding to $t$, and all other 
  coordinates equal to 0.

We define the (weighted) \emph{Laplacian $\LL$ of $G$ \mbox{\rm (with respect to the resistances $\rr$)}} to be the $n \times n$ matrix
  $$\LL:= \BB \CC \BB^T,$$
  where $\CC$ is the $m\times m$ diagonal matrix with $\CC_{e,e}=c_e=1/r_e$.
One can easily check that its entries are given by
\[
  \LL_{u,v} = \begin{cases}
\sum_{e\in E^{+}(u)\cup E^{-}(u)} c_e & \text{if $u=v$,}\\
-c_{e}  & \text{if $e=(u,v)$ is an edge of $G$, and}\\
0 & \text{otherwise.}
\end{cases}
\]

Let $\RR=\CC^{-1}$ be the diagonal matrix with $\RR_{e,e}=r_e$.  
The energy of a flow $\ff$ is given by
  $$  \energy{\rr}{f}:=\sum_{e} r_{e} \ff (e)^2 = \ff^T\RR\ff =\norm{\RR^{1/2} \ff}^2. $$
The electrical flow of value 1 thus corresponds to the vector $\ff$ that minimizes
 $\norm{\RR^{1/2} \ff }^2 $ subject to
  $\BB \ff = \bchist.$
If $f$ is an electrical flow, it is well known that it is a \textit{potential flow}, which means that there
  is a vector $\phiphi \in \Reals{V}$ such that
\[
  \ff(u,v) = \frac{\phi_{v} - \phi_{u}}{r_{u,v}}.  
\]
That is,
\[
  \ff =  \CC \BB^T \phiphi = \RR^{-1}\BB^T \phiphi.
\]
Applying $\BB \ff = \bchist$, we have 
$\BB\ff = \BB \CC \BB^T \phiphi = \bchist$, and hence
 the vertex potentials are given by
\[
  \phiphi=\pinv{\LL}\bchist,
\]
where $\pinv{\LL}$ denotes the Moore-Penrose pseudo-inverse of $\LL$.
Thus, the electrical flow $\ff$ is given by the expression  
\[
  \ff = \CC\BB^T \pinv{\LL}\bchist \label{flowequation}.
\]
This lets us rewrite the energy of the electrical flow of value 1 as
\begin{equation}\label{eq:alt_energy}
 \energy{\rr}{\ff }=
   \ff^T\RR\ff =
  \left(\bchist^T {\pinv{\LL}}^T \BB \CC^T \right) \RR \left(\CC\BB^T \pinv{\LL}\bchist \right) = 
  \bchist \pinv{\LL }\LL  \pinv{\LL } \bchist = \bchist^T \pinv{\LL } \bchist =
  \phiphi^{T} \LL  \phiphi .
 \end{equation}

\subsubsection{Effective $s$-$t$ Resistance and Effective $s$-$t$ Conductance}

Our analysis will make repeated use of two basic quantities from the theory of electrical networks, the effective $s$-$t$ resistance and effective $s$-$t$ conductance.

Let $f$ be the electrical $s$-$t$ flow of value 1, and let $\phiphi$ be its vector of vertex potentials.
 The \emph{effective $s$-$t$ resistance of $G$ with respect to the resistances $\rr$} is given by 
 $$\Reff{\rr} = \phi(s)-\phi(t).$$
    Throughout this paper, we will only look at the effective resistance between the vertices $s$ and $t$ in the graph $G$, so we
   suppress these letters in our notation and simply write $\Reff{\rr}$.

Using our linear algebraic description of the electrical flow and Equation~\eqref{eq:alt_energy}, we have
$$\Reff{\rr} = \phi(s)-\phi(t) 
   = \bchist^{T} \phiphi 
  = \bchist^{T} \pinv{\LL}\bchist 
  = \energy{\rr}{\ff }.$$
This gives us an alternative description of the effective $s$-$t$ resistance as the energy of the electrical flow of value 1.

It will sometimes be convenient to use the related notion of the \emph{effective $s$-$t$ conductance of $G$ with respect to the resistances $\rr$}, which we define by
$$\Ceff{\rr} = 1/\Reff{\rr}.$$
We note that this is equal to the value of the electrical flow in which $\phi(s)-\phi(t)=1$.

\subsubsection{Approximately Computing Electrical Flows}

From the algorithmic point of view, the crucial property of the
  Laplacian $\LL$ is that it is symmetric and \emph{diagonally
  dominant}, i.e., for any $u$, 
  $\sum_{v'\neq u} |\LL_{u,v}|\leq  \LL_{u,v}$. 
This allows us to use the result of Koutis, Miller, and Peng
  \cite{KoutisMP10}, which builds on the work of Spielman and Teng
  \cite{SpielmanTeng_solving},  to approximately
  solve our linear system in nearly-linear time. 
By rounding the approximate solution to a flow, we can
  prove the following theorem (see Appendix~\ref{sec:linsolve}
  for a proof).

\begin{theorem}[Fast Approximation of Electrical Flows]\label{thm:computing_electrical_flow}
For any $\delta >0$, any $F>0$, and any vector $\rr$ of resistances
  in which the ratio of the largest to the smallest resistance is at most $R$, 
  we can compute, in time $\Ot{m\log R/\delta}$, 
  a vector of vertex potentials $\tphiphi$
  and an $s$-$t$ flow $\tf$ of value $F$ such that
\begin{itemize}
\item [a.] $\energy{\rr}{\tf} \leq (1 + \delta) \energy{\rr}{f}$,
  where $f$ is the electrical $s$-$t$ flow of value $F$, and
\item [b.] for every edge $e$,
\[
\abs{r_{e} f_{e}^{2} - r_{e} \tf_{e}^{2}} 
  \leq \frac{\delta}{2 m R}    \energy{\rr}{f},
\]
where $f$ is the true electrical flow.  
\item [c.]
\[
  \tphi_{s} - \tphi_{t} \geq \left(1 - \frac{\delta}{12 n m R} \right) F \Reff{\rr}.
\]
\end{itemize}
We will refer to a flow meeting the above conditions as a \emph{$\delta$-approximate electrical flow}.
\end{theorem}    

\subsection{How the Resistance of an Edge Influences the Effective Resistance}

In this section, we will study how changing the resistance of an edge affects the effective resistance of the graph.
This will be a key component of the analysis of our $\Ot{m^{4/3}\cdot\mathrm{poly}(1/\epsilon)}$ algorithms.
We will make use of the following standard fact about effective conductance; we refer the reader to \cite[Chapter IX.2, Corollary 5]{Bollobas98} for a proof. 
\begin{fact}\label{fa:effective_cond_express}
For any $G=(V,E)$ and any vector of resistances $\rr$, 
  $$
    \Ceff{\rr}=  \min_{\phiphi\, | \, \phi_{s}=1,\atop 
   \phi_{t}=0} \sum_{(u,v)\in E} \frac{(\phi_{u}-\phi_{v})^{2}}{r_{(u,v)}}. 
  $$
Furthermore, the equality is attained for $\phiphi$ 
  being vector of vertex potentials corresponding to the electrical $s$-$t$ flow of $G$ 
  (with respect to $\rr$) of value  $1/\Reff{\rr}$. 
\end{fact}

\begin{corollary}[Rayleigh Monotonicity]  \label{cor:rayleigh_mon}
If $r'_e \geq r_e$ for all $e \in E$, then $\Reff{\rr'} \geq \Reff{\rr}$.
\end{corollary}
\begin{proof}
For any $\phiphi$, 
$$ \sum_{(u,v)\in E} \frac{(\phi_{u}-\phi_{v})^{2}}{r'_{(u,v)}} \leq  \sum_{(u,v)\in E} \frac{(\phi_{u}-\phi_{v})^{2}}{r_{(u,v)}}, $$
so the minima of these expressions over possible  values of $\phi$ obey the same relation, and thus $\Ceff{\rr'} \leq \Ceff{\rr}$.  Inverting both sides of this inequality yields the desired result.
\end{proof}

Our analysis of the $\Ot{m^{4/3}}$ algorithm will require the following lemma, which gives a lower bound on the effect that increasing the resistance of an edge can have on the effective resistance.
\begin{lemma}\label{lem:effect_of_res_incr}
Let $f$ be an electrical $s$-$t$ flow on a graph $G$ with resistances $\rr$.  Suppose that some edge $h=(i,j)$ accounts for a $\beta$ fraction of the total energy of $f$, i.e.,
$$f(h)^2 r_h = \beta\energy{\rr}{f}. $$
For some $\gamma>0$, define new resistances $\rr'$ such that $r_h' = \gamma r_h$, and $r_e'=r_e$ for all $e\neq h$.  Then
$$\Reff{\rr'}  \geq \frac{\gamma}{\beta+\gamma(1-\beta)}\Reff{\rr}.$$
In particular:
\begin{itemize}
\item If we ``cut'' the edge $h$ by setting $\gamma=\infty$, then
$$\Reff{\rr'}  \geq \frac{\Reff{\rr}}{1-\beta}.$$
\item If we slightly increase the effective resistance of the edge $h$ by setting
  $\gamma = (1 + \epsilon)$ with $\epsilon \leq 1$, then
\[
  \Reff{\rr'} \geq \frac{1+\epsilon}{\beta+(1+\epsilon)(1-\beta)}\Reff{\rr} 
  \geq \left(1 + \frac{\epsilon \beta}{2} \right)  \Reff{\rr}.
\]
\end{itemize}
\end{lemma}
\begin{proof}
\newcommand{\phif}{\phi^{f}} 
\newcommand{\phiphif}{\boldsymbol{\mathit{\phi}}^{f}}
The assumptions of the theorem are unchanged if we multiply $f$ by a constant, so we may assume
without loss of generality that $f$ is the electrical $s$-$t$ flow of value $1/\Reff{\rr}$.  If $\phiphif$ is the vector of vertex potentials corresponding to $f$, this gives $\phif_s-\phif_t=1$.   Since adding a constant to the potentials doesn't change the flow,  we may assume that $\phi_s=1$ and $\phi_t=0$.
By Fact~\ref{fa:effective_cond_express},
\begin{align*} \Ceff{\rr}&= \sum_{(u,v)\in E} \frac{\big(\phif_{u}-\phif_{v}\big)^{2}}{r_{(u,v)}}
\quad=\quad \frac{\big(\phif_{i}-\phif_{j}\big)^{2}}{r_{h}} 
     +  \sum_{(u,v)\in E\setminus \{h\}} \frac{\big(\phif_{u}-\phif_{v}\big)^{2}}{r_{(u,v)}}.
\end{align*}
The assumption that $h$ contributes a $\beta$ fraction of the total energy implies that, in the above expression,
$$\frac{\left(\phif_{i}-\phif_{j}\right)^{2}}{r_{h}} = \beta \Ceff{\rr},$$
and thus
$$ \sum_{(u,v)\in E\setminus \{h\}} \frac{\left(\phif_{u}-\phif_{v}\right)^{2}}{r_{(u,v)}}=(1-\beta)\Ceff{\rr}.$$

We will obtain our bound on $\Ceff{\rr'}$ by plugging the original vector of potentials $\phif$ into the expression in Fact~\ref{fa:effective_cond_express}:
\begin{align*} \Ceff{\rr'}&=\min_{\phiphi\, | \, \phi_{s}=1,\atop 
   \phi_{t}=0} \sum_{(u,v)\in E} \frac{\big(\phi_{u}-\phi_{v}\big)^{2}}{r'_{(u,v)}} 
\quad\leq  \quad
    \sum_{(u,v)\in E} \frac{\big(\phif_{u}-\phif_{v}\big)^{2}}{r'_{(u,v)}}\\
&= \frac{\big(\phif_{i}-\phif_{j}\big)^{2}}{r'_{h}} 
     +  \sum_{(u,v)\in E\setminus \{h\}} \frac{\big(\phif_{u}-\phif_{v}\big)^{2}}{r'_{(u,v)}}
\quad =\quad \frac{\big(\phif_{i}-\phif_{j}\big)^{2}}{\gamma r_{h}} 
    +  \sum_{(u,v)\in E\setminus \{h\}} \frac{\big(\phif_{u}-\phif_{v}\big)^{2}}{r_{(u,v)}}\\
&= \frac{\beta}{\gamma} \, \Ceff{\rr} + (1-\beta)\Ceff{\rr} 
   \quad =\quad \Ceff{\rr}\left( \frac{\beta +\gamma(1-\beta)}{\gamma}\right).
\end{align*}
Since $\Reff{\rr}=1/\Ceff{\rr}$ and $\Reff{\rr'}=1/\Ceff{\rr'}$, the desired result follows.
\end{proof}

\newcommand{\fix}{\colorbox{yellow}{FIX}}
\section{A Simple \texorpdfstring{$\Ot{m^{3/2}\eps^{-5/2}}$}{\~O(m\textasciicircum(3/2) eps\textasciicircum(-5/2))}-Time Flow Algorithm}
\label{sec:simple}
Before describing our $\Ot{m^{4/3}\eps^{-3}}$ algorithm, we will describe a simpler 
  algorithm that finds a $(1-\epsilon)$-approximately maximum flow in time $\Ot{m^{3/2}\eps^{-5/2}}$.  
Our final algorithm will be obtained by carefully modifying the one described here.

The algorithm will employ the multiplicative weights update method \cite{AroraHK05,PlotkinShmoysTardos}.  
In our setting, one can understand the multiplicative weights method as a way of taking an algorithm that solves a flow problem very 
  crudely and, by calling it repeatedly, converts it into an algorithm that gives a good approximation for the maximum flow in $G$.   
The crude algorithm is called as a black-box, so it can be thought of as an oracle that answers a certain type of query.

In this section, we provide a self-contained description of the multiplicative weights method when it is specialized to our setting.  
In Section~\ref{sec:MW_overview}, we will describe the requirements on the oracle, give an algorithm that iteratively uses it to obtain a  
  $(1-\epsilon)$-approximately maximum flow, and state how the number of iterations required by the algorithm depends on the properties of the oracle. 
In Section~\ref{sec:sqrtn_oracle}, we will describe how to implement the oracle using electrical flows.  
Finally, in Section~\ref{sec:convergence_of_mw} we will provide a simple proof of the convergence bound set forth in Section~\ref{sec:MW_overview}.

\subsection{Multiplicative Weights Method: From Electrical Flows
  to Maximum Flows}
\label{sec:MW_overview}

For an $s$-$t$ flow $f$, we define the \emph{congestion} of an edge $e$ to be the ratio
  $$\cong{f}{e}:=\frac{|f(e)|}{u_{e}}$$ 
  between the flow on an edge and its capacity.  
In particular, an $s$-$t$ flow is feasible  if and only if
   $\cong{f}{e}\leq 1$ for all $e\in E$.

The multiplicative weights method will use a subroutine that we will 
   refer to as an  \emph{$(\eps,\rho)$-oracle}.  
This oracle will take as input a number $F$ and a vector $\ww$ of edge weights.  
For any $F \leq \maxflow$, we know that there exists a way to route $F$ units 
  of flow in $G$ so that all of the edge capacities are respected.  
Our oracle will provide a weaker guarantee:
When $F\leq \maxflow$, it will satisfy \emph{all} of the capacity constraints up to a multiplicative factor
of  $\rho$,\footnote{Up to polynomial factors in $1/\epsilon$, the value of $\rho$ will be  $\Theta(\sqrt{m})$ in this section, and $\widetilde{\Theta}(m^{1/3})$ later in the paper.} and it will satisfy the \emph{average} of these constraints, weighted by the $w_i$, up to a (much better) multiplicative factor
of $(1+\epsilon)$.   
When $F>\maxflow$, the oracle will either output an $s$-$t$ flow satisfing the conditions above,
   or it will return $\fail$.

Formally, we will use the following definition:
\begin{definition}[$(\eps,\rho)$ oracle]\label{def:maxflow_oracle}
For $\eps>0$ and $\rho>0$, an \emph{$(\eps,\rho)$ oracle}
  is an algorithm 
  that, given 
  a real number $F>0$ and a vector $\ww$ of edge
  weights with $w_{e}\geq 1$ for all $e$,
   returns an $s$-$t$ flow $f$ such that:
\begin{enumerate}
\item If $F\leq \maxflow$, then it outputs an $s$-$t$ flow $f$ satisfying:
	\begin{enumerate}[(i)]
	\item $\thr{f}=F$;\label{cond:throughput}
	\item $\sum_{e} w_{e} \cong{f}{e}\leq (1+\eps)|\ww|_{1}$, where $|\ww|_{1}:=\sum_{e}w_{e}$;\label{cond:weight}
	\item $\max_{e} \cong{f}{e}\leq \rho$.\label{cond:rho}
	\end{enumerate}
\item If $F>\maxflow$, then it either outputs a flow $f$ satisfying
  conditions (i), (ii), (iii) or outputs $\fail$.
\end{enumerate} 
\end{definition}

Our algorithm will be given a flow value $F$ as an input.  If $F\leq \maxflow$, it will return a flow of value at least $(1-O(\epsilon))F$.  If $F>\maxflow$, it will either return a flow of value at least $(1-O(\epsilon))F$ (which may occur if $F$ is only slightly greater than $\maxflow$) or it will return $\fail$.  This allows us to find a $(1-O(\epsilon))$-approximation of $\maxflow$ using binary search.  As outlined in Section~\ref{sec:max_flow_prelim}, we can obtain a crude bound $B$ in time $O(m+n\log n)$ such that $B \leq \maxflow \leq mB$, so the  binary search will only call our algorithm a logarithmic number of times.

In Figure~\ref{fig:multiplicative_alg}, we present our simple
  algorithm, which applies the multiplicative weights update routine
  to approximate the maximum flow by calling an $(\epsilon,\rho)$-flow
  oracle.
The algorithm initializes all of the weights to $1$ and 
   calls the oracle with these weights.  
The call returns a flow that satisfies conditions (i), (ii), and (iii)
  defined above. 
It then multiplies the weight of each edge $e$ 
  by $(1+\frac{\eps}{\rho}\cong{f^{i}}{e})$. 
Note that if the congestion of an edge is poor, say close to $\rho$, 
  then its weight will increase by a factor of $(1+\epsilon)$.
On the other hand, if the flow on an edge is no more than its
  capacity,
  then the new weight of the edge is essentially unchanged. 
This will put a larger fraction of the weight on the violated
  constraints, so the oracle will be forced to return a solution that comes closer to satisfying them
(possibly at the expense of other edges).
In the end, we return the average of all of the flows as our answer.

\IncMargin{0.5em}
\begin{algorithm}
\DontPrintSemicolon
\SetAlFnt{\small \sf}
\AlFnt
\SetKw{Return}{return}
\SetKwInOut{Input}{Input}\SetKwInOut{Output}{Output}
\Input{A graph $G=(V,E)$ with capacities $\{u_{e}\}_{e}$, a target flow value $F$, and an $(\eps,\rho)$-oracle $\oracle$}
\Output{Either a flow $\of$,  or $\fail$ indicating that $F>\maxflow$;}
\BlankLine
{Initialize $w_{e}^{0}\leftarrow 1$ for all edges $e$, and $N\leftarrow \frac{2\rho\ln m }{\eps^{2}}$}\;
\For{$i:=1,\ldots,N$}{
Query $\oracle$ with edge weights given by $\ww^{i-1}$ and target flow value $F$\;
\lIf{\AlFnt{$\oracle$ returns $\fail$}}{\AlFnt{\Return $\fail$}}\;
\Else{
Let $f^{i}$ be the returned flow\;
$w_{e}^{i}\leftarrow w_{e}^{i-1}(1+\frac{\eps}{\rho}\cong{f^{i}}{e})$ for each $e\in E$\;
}
}
\Return $\of\leftarrow \frac{(1-\eps)^2}{(1+\eps)N}(\sum_{i} f^{i})$\;

\caption{Multiplicative-weights-update routine}\label{fig:multiplicative_alg}
\end{algorithm}\DecMargin{0.5em}

The key point in analyzing this algorithm is that the total weight on $G$ does not grow too quickly, due to the 
average congestion constraint on the flows returned by the oracle $\oracle$.
However,
   if an edge $e$ consistently suffers large congestion
   in a sequence of flows returned by $\oracle$, then its  
    weight increases rapidly relative to the total weight, which will significantly penalize any further 
    congestion the edge in the subsequent flows.
If this were to occur too many times, its weight would exceed the total weight, which obviously cannot occur.
In Section~\ref{sec:convergence_of_mw}, we will prove the
  following theorem by showing that our algorithm converges 
  in $2\rho \ln m/\epsilon^2$ iterations.
\begin{theorem}[Approximating Maximum Flows by Multiplicative Weights]\label{thm:multiplicative}
For any $0 < \eps < 1/2$ and $\rho>0$, 
  given an $(\eps,\rho)$-flow oracle with running time 
  $T(m,1/\eps,U)$,  one can obtain an algorithm that computes
  a $(1-O(\eps))$-approximate maximum flow in a capacitated,
  undirected graph in time  $\Ot{\rho\eps^{-2}\cdot T(m,1/\eps,U)}$.
\end{theorem}

Note that the number of iterations of the algorithm above
   grows linearly with the value of $\rho$, 
   which we call the $\emph{width}$ of the oracle.
Intuitively, this should be necessary because the final flow across an edge $e$
  is equal to the average of the flows sent over it by 
  all $f^{i}$.
If we send $\rho\cdot u_e$ units of flow across the edge $e$ in some step,
  then we will need at least $\Omega(\rho)$ iterations 
  to drop the average to $1$.%
\footnote{ Strictly speaking, it is possible that we could do better than this by sending a
  large amount of flow across the edge in the opposite direction.  However, nothing in our algorithm 
  aims to obtain this kind of cancellation, so we shouldn't expect to be able to systematically exploit it.}

\subsection{Constructing an Oracle of Width \texorpdfstring{$3\sqrt{m/\eps}$}{3 sqrt(m/epsilon)} Using Electrical Flows}\label{sec:sqrtn_oracle}
Given Theorem \ref{thm:multiplicative},
 our problem is thus reduced to designing 
 an efficient oracle that has a small width.
In this subsection, we will give simple $\Ot{m \log \epsilon^{-1} }$
  time implementation of an $(\eps, 3\sqrt{m/\eps})$-flow oracle for any $0<\epsilon<1/2$.
By Theorem \ref{thm:multiplicative}, this will immediately yield an 
  $\Ot{m^{3/2}\eps^{-5/2}}$ time algorithm for finding an 
  approximately maximum flow. 

To build such an oracle, we set 
\begin{equation}\label{eq:resistances_oracle}
r_{e}:=\frac{1}{u_e^2}\left(w_{e}+\frac{\eps |\ww|_{1}}{3m} \right)
\end{equation}
 for each edge $e$, and we use the procedure from 
 Theorem~\ref{thm:computing_electrical_flow} to  approximate the electrical flow 
 that sends $F$ units of flow from $s$ to $t$ in a network 
 whose resistances are given by the $r_e$.   
The pseudocode for this oracle is shown in Figure~\ref{fig:oracle_simple}.

\IncMargin{0.5em}
\begin{algorithm}
\DontPrintSemicolon
\SetAlFnt{\small \sf}
\AlFnt
\SetKw{Return}{return}
\SetKwInOut{Input}{Input}\SetKwInOut{Output}{Output}
\Input{A graph $G=(V,E)$ with capacities $\{u_{e}\}_{e}$, a target
  flow value $F$, and edge weights  $\{w_e\}_e$} 
\Output{Either a flow $\tf$, or 
  $\fail$ indicating that $F>\maxflow$}
\BlankLine
$r_{e}\leftarrow \frac{1}{u_e^2}\left(w_{e}+\frac{\eps |\ww|_{1}}{3m} \right)$ for each $e\in E$\;

Find an $(\epsilon/3)$-approximate electrical flow $\tf$ using Theorem \ref{thm:computing_electrical_flow} on $G$ with resistances $\rr$ and target \\  \quad flow value $F$ \;
\lIf{\AlFnt{$\energy{\rr}{\tf}>(1+\eps)|\ww|_{1}$}}{\AlFnt{\Return{$\fail$}}}\;
\lElse{\Return $\tf$}\;
\caption{A simple implementation of an $\left(\epsilon, 3\sqrt{m/\epsilon}\right)$ oracle}\label{fig:oracle_simple}
\end{algorithm}\DecMargin{0.5em}

We now show that the resulting flow $\tf$ has the properties required by
Definition~\ref{def:maxflow_oracle}.  
Since $|\tf|=F$ by construction, we only need to demonstrate 
  the bounds on the average congestion (weighted by the $w_e$) and 
  the maximum congestion. 
We will use the basic fact that electrical flows minimize
  the energy of the flow.
Our analysis will then compare the energy of $\tf$ with that of an optimal
  max flow.
Intuitively, the $w_e$ term in  
  Equation~\eqref{eq:resistances_oracle} guarantees the bound on the
  average congestion, while the $\eps |\ww|_{1}/(3m)$ term guarantees
  the bound on the maximum congestion.

Suppose $f^{*}$ is a maximum flow.
By its feasibility, $\cong{f^{*}}{e}\leq 1$ for all $e$, so
\begin{align*}
\energy{\rr}{f^{*}}
&=\sum_{e} \left(w_{e}+\frac{\eps |\ww|_{1}}{3m}\right)\left(\frac{f^{*}(e)}{u_{e}}\right)^{2}\\
&=\sum_{e} \left(w_{e}+\frac{\eps |\ww|_{1}}{3m}\right)\left(\cong{f^{*}}{e}\right)^{2}\\
&\leq \sum_{e} \left( w_{e}+\frac{\eps |\ww|_{1}}{3m} \right)\\
&=\left(1+\frac{\eps}{3}\right)|\ww|_{1}.\\ 
\end{align*}

Since the electrical flow minimizes the energy, $\energy{\rr}{f^{*}}$ is an upper bound on the energy of the electrical flow of value $F$ whenever $F\leq \maxflow$.
In this case, Theorem~\ref{thm:computing_electrical_flow} implies that our $(\epsilon/3)$-approximate electrical flow satisfies
\begin{equation}\label{eq:approx_flow_energy_bound}
 \energy{\rr}{\tf}\leq \left(1+\frac{\eps}{3}\right)\energy{\rr}{f^{*}}\leq \left(1+\frac{\eps}{3}\right)^{2}|\ww|_{1}
\leq \left(1+\eps\right)|\ww|_{1}.
\end{equation}
This shows that our oracle will never output $\fail$ when $F\leq \maxflow$. 
It thus suffices to show that the energy bound $\energy{\rr}{\tf}>(1+\eps)|\ww|_{1}$, which holds whenever the algorithm does not return $\fail$, implies the required bounds on the average and worst-case congestion.  To see this, we note that the energy bound implies that 
\begin{equation}\label{eq:avg_bound} 
\sum_{e} w_{e}\left(\cong{\tf}{e}\right)^{2}\leq\left(1+{\eps}\right)|\ww|_{1},
\end{equation}
and, for all $e\in E$,
\begin{equation}\label{eq:max_bound}
\frac{\eps |\ww|_{1}}{3m}\left(\cong{\tf}{e}\right)^{2}\leq \left(1+\eps\right)|\ww|_{1}.
\end{equation}
By the Cauchy-Schwarz inequality,
\begin{equation}\label{eq:avg_bound2}
\left(\sum_e w_e \cong{\tf}{e}\right)^2\leq 
  |\ww|_{1} \left(\sum_e w_e\left(\cong{\tf}{e}\right)^{2}\right),
\end{equation}
  so Equation~\eqref{eq:avg_bound} gives us that
 \begin{equation}\label{eq:avg_bound3}
 \sum_e w_e \cong{\tf}{e} \leq \sqrt{1+\eps}\, |\ww|_{1}<  \left(1+\eps\right)|\ww|_{1},
 \end{equation}
  which is the required bound on the average congestion.  
Furthermore, Equation~\eqref{eq:max_bound} and the fact that $\epsilon < 1/2$ implies that

\[
\cong{\hf}{e}
 \leq \sqrt{\frac{3m(1+\epsilon)}{\epsilon}} 
 \leq  3 \sqrt{m/\eps}
\]
  for all $e$, which establishes the required bound on the maximum congestion.  
So our algorithm implements an $(\eps, 3 \sqrt{m/\eps})$-oracle, as desired.

To bound the running time of this oracle, 
  recall that we can assume all edge capacities lie between
  $1$ and $U = m^{2}/ \epsilon$ and compute
\begin{equation}\label{eq:R_bound}
R=\max_{e,e'} \frac{r_{e}}{r_{e'}}
\leq 
U^{2} \frac{\max_{e} r_{e} + \epsilon \abs{\ww}_{1}}{ \epsilon \abs{\ww}_{1}} 
\leq 
U^{2} \frac{m + \epsilon}{\epsilon }=O\left((m/\eps)^{O(1)}\right).
\end{equation}

This establishes an upper bound on the ratio
  of the largest resistance to the smallest resistance.
Thus, by Theorem \ref{thm:computing_electrical_flow}, the running time of
this implementation is $\Ot{m\log R/\eps}=\Ot{m \log 1/\eps}$.  
Combining this with Theorem \ref{thm:multiplicative}, we have shown
\begin{theorem}\label{thm:result_basic}
  For any $0<\eps<1/2$, the maximum flow problem can be $(1-\eps)$-approximated in $\Ot{m^{3/2}\eps^{-5/2}}$ time.
\end{theorem}


\subsection{The Convergence of Multiplicative Weights}
\label{sec:convergence_of_mw}

In this section, we prove Theorem \ref{thm:multiplicative} 
  by analyzing the multiplicative weights update
  algorithm shown in Figure~\ref{fig:multiplicative_alg}.

Our analysis will be based on the potential function
$\mu_{i}:=|\ww^{i}|_{1}$.  
Clearly, $\mu_{0}=m$ and this potential
 only increases during the course of the algorithm. 
It follows from  condition (1.ii)
  of Definition \ref{def:maxflow_oracle}
  that if we run the $(\eps,\rho)$-oracle $\oracle$ with $F =
  \maxflow$,
  then 

\begin{equation}\label{eq:oracle_cond_weight}
\sum_{e} w_{e}^{i} \cong{f^{i+1}}{e}\leq (1+\eps)|\ww^{i}|_{1}, \ \ \text{for all $i\geq 1$}.
\end{equation}
By condition (1.iii) of Definition \ref{def:maxflow_oracle},
\begin{equation}\label{eq:oracle_cond_rho}
\cong{f^{i}}{e}\leq \rho, \ \ \text{for all $i\geq 1$ and any edge $e$}.
\end{equation}

We start by upper bounding the total growth of $\mu_{i}$ thoughout the algorithm.

\begin{lemma}\label{lem:upper_sum_potential}
For any $i\geq 0$, $$\mu_{i+1}\leq \mu_{i}\exp\left(\frac{(1+\eps)\eps}{\rho}\right).$$ 
In particular, $\|\ww^{N}\|_{1}=\mu_{N}\leq m\exp\left(\frac{(1+\eps)\eps}{\rho}N\right)=n^{O(1/\eps)}$.
\end{lemma}

\begin{proof}
For any $i\geq 0$, we have

\[
\mu_{i+1}=\sum_{e} w_{e}^{i+1}=\sum_{e} w_{e}^{i}\left(1+\frac{\eps}{\rho} \cong{f^{i+1}}{e}\right)=\sum_{e} w_{e}^{i}+ \frac{\eps}{\rho}\sum_{e} w_{e}^{i} \cong{f^{i+1}}{e}\leq \mu_{i}+\frac{(1+\eps)\eps}{\rho}|\ww^{i}|_{1}, 
\] 
where the last inequality follows from \eqref{eq:oracle_cond_weight}. 
Thus, we can conclude that

\[
\mu_{i+1}\leq \mu_{i}+\frac{(1+\eps)\eps}{\rho}|\ww^{i}|_{1}=\mu_{i}(1+\frac{(1+\eps)\eps}{\rho})\leq \mu_{i}\exp\left(\frac{(1+\eps)\eps}{\rho}\right),
\]
as desired. The lemma follows.
\end{proof}

One of the consequences of the above lemma is that whenever we make a call to the oracle, the total weight $\|\ww^{i}\|_{1}$ is at most $n^{O(1/\eps)}$. 
Thus, the running time of our algorithm is $\Ot{\rho\eps^{-2}\cdot T(m,1/\eps,U,m^{O(1/\eps)})}$ as claimed.

Next, we bound the final weight $w_{e}^{N}$ of a particular edge $e$ with the congestion $\cong{\of}{e}$ that this edge suffers in our final flow $\of$.

\begin{lemma}\label{lem:single_edge_potential}
For any edge $e$ and $i\geq 0$, $$w_{e}^{i}\geq \exp\left(\frac{(1-\eps)\eps}{\rho} \sum_{j\geq 1}^{i} \cong{f^{j}}{e}\right).$$ 
In particular,
$w_{e}^{N}\geq \exp\left(\frac{(1+\eps)\eps N}{(1-\eps)\rho} \cong{\of}{e}\right)$.
\end{lemma}

\begin{proof}
For any $i\geq 0$, we have
\[
w_{e}^{i}=\prod_{j\geq 1}^{i} \left(1+\frac{\eps}{\rho}\cong{f^{j}}{e}\right)\geq \prod_{j\geq 1}^{i} \exp\left(\frac{(1-\eps)\eps}{\rho}\cong{f^{j}}{e}\right),
\] 
where we used \eqref{eq:oracle_cond_rho} and that for any $1/2>\eps>0$ and $x\in [0,1]$:
\[
\exp((1-\eps)\eps x)\leq (1-\eps x).
\]

Now, the lemma follows since
\[
w_{e}^{i}\geq \prod_{j\geq 1}^{i} \exp\left(\frac{(1-\eps)\eps}{\rho}\cong{f^{j}}{e}\right) = \exp\left(\frac{(1-\eps)\eps}{\rho} \sum_{j\geq 1}^{i} \cong{f^{j}}{e}\right),
\]
and for $i=N$
\[
w_{e}^{N}\geq \exp\left(\frac{(1-\eps)\eps}{\rho} \sum_{j\geq 1}^{N} \cong{f^{j}}{e}\right)=\exp\left(\frac{(1+\eps)\eps N}{(1-\eps)\rho} \cong{\of}{e}\right).
\]
\end{proof}
Finally, by Lemmas \ref{lem:upper_sum_potential} and \ref{lem:single_edge_potential}, we conclude that  for any edge $e$,
\[
m\exp\left(\frac{(1+\eps)\eps N}{\rho}\right)\geq \mu_{N}=|\ww^{N}|_{1}\geq w_{e}^{N}\geq \exp\left(\frac{(1+\eps)\eps N}{(1-\eps)\rho} \cong{\of}{e}\right).
\]
This implies that
\[
\cong{\of}{e}\leq 1-\eps + \frac{(1-\eps) \rho \ln m}{(1+\eps)\eps N}=1-\eps + \frac{\eps(1-\eps)}{2(1+\eps)}\leq 1
\]
  for every edge $e$. 
Thus, we see that $\of$ is a feasible $s$-$t$ flow and, since each $f^{i}$ has throughput $\maxflow$, the throughput $\thr{\of}$ of $\of$   
  is $\frac{(1-\eps)^{2}}{(1+\eps)} \maxflow\geq (1-O(\eps))\maxflow$ for $1/2>\eps>0$, as desired.


\section{An \texorpdfstring{$\Ot{mn^{1/3}\eps^{-11/3}}$}{\~O(mn\textasciicircum(1/3) eps\textasciicircum(-11/3))} Algorithm for Approximate Maximum Flow}
\label{sec:improved_flow}

In this section, we modify our algorithm to run in time $\Ot{m^{4/3}\eps^{-3}}$. We then combine this with the smoothing and sampling techniques of Karger \cite{Karger98} to obtain an $\Ot{mn^{1/3}\eps^{-11/3}}$-time algorithm.  

For fixed $\epsilon$, the algorithm in the previous section required us to compute $\Ot{m^{1/2}}$ 
  electrical flows, each of which took time $\Ot{m}$, which led to a running time of $\Ot{m^{3/2}}$.  
To reduce this to $\Ot{m^{4/3}}$, we'll show how to find an approximate flow while computing only $\Ot{m^{1/3}}$ electrical flows.

Our analysis of the oracle from Section~\ref{sec:sqrtn_oracle} was fairly simplistic, and one might hope to improve the running time 
  of the algorithm by proving a tighter bound on the width.
Unfortunately, the graph in Figure~\ref{badgraph} shows that our analysis was essentially tight.  
The graph consists of $k$ parallel paths of length $k$ connecting $s$ to $t$, along  with a single
  edge $e$ that directly connects $s$ to $t$.  
The max flow in this graph is $k+1$.  
In the first call made to the oracle by the multiplicative weights routine, all of the edges will
  have the same resistance.  
In this case, the electrical flow of value $k+1$ will send $(k+1)/2k$ units of flow along each of the $k$ paths 
  and $(k+1)/2$ units of flow across $e$.  
Since the graph has $m=\Theta(k^2)$, the width of the oracle in this case is $\Theta(m^{1/2})$.

\begin{figure}[htbp]
\begin{center}
\includegraphics[height=3in]{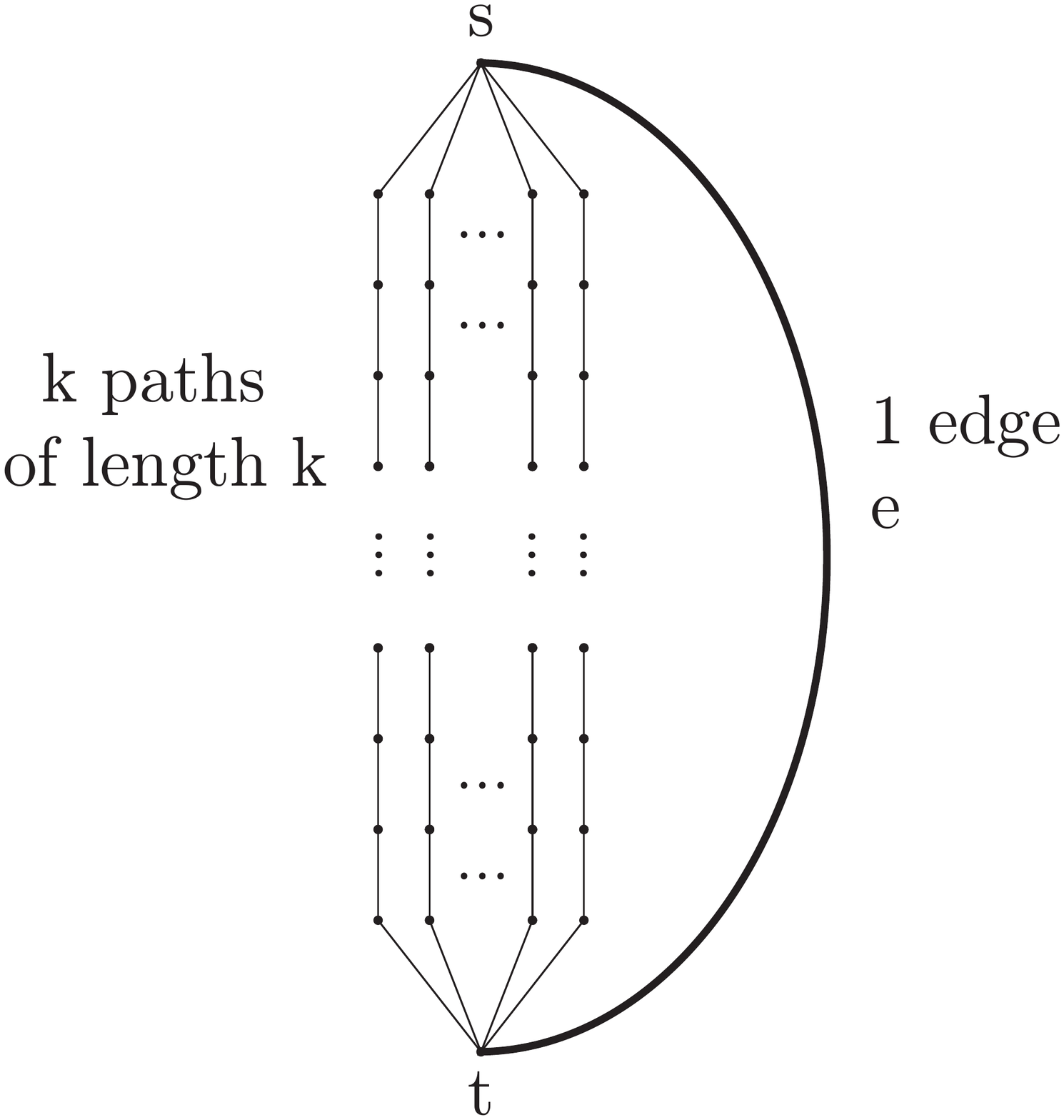}
\caption{A graph on which the electrical flow sends approximately $\sqrt{m}$ units of flow across an edge when sending the maximum flow $\maxflow$ from $s$ to $t$.}
\label{badgraph}
\end{center}
\end{figure}

\subsection{The Improved Algorithm}

The above example shows that it is possible for the electrical flow returned by the oracle to exceed the 
  edge capacities by  $\Theta(m^{1/2})$.  
However, we note that if one removes the edge $e$ from the graph in Figure~\ref{badgraph}, the electrical flow on the 
  resulting graph is much better behaved, but the value of the maximum flow is only very slightly reduced.  This demonstrates a phenomenon that will be central to our improved algorithm:   while instances in which the electrical flow sends a huge amount of flow over some edges exist, they are somewhat fragile, and they are often drastically improved by removing the bad edges.

This motivates us to modify our algorithm as follows.  We'll set $\rhop$ to be some value smaller than the actual worst-case bound of $\Ot{m^{1/2}}$.  (It will end up being $\Ot{m^{1/3}}$.)
The oracle will begin by computing an electrical flow as before.  
However, when this electrical flow exceeds the capacity of some edge $e$ by a factor greater than $\rhop$, we'll 
  remove $e$ from the graph and try again, keeping all of the other weights the same.  
We'll repeat this process until we obtain a flow in which all edges flow at most a factor of $\rhop$ times their capacity
  (or some failure condition is reached), and we'll use this flow in our multiplicative weights routine.  
When the oracle removes an edge, it is added to a set $H$ of \emph{forbidden edges}.  
These edges will be permanently removed from the graph, i.e., they will not be included in the graphs supplied to 
  future invocations of the oracle.

In Figures~\ref{fig:improved_oracle} and~\ref{fig:improved_alg}, we present the modified versions of the oracle and overall algorithm, where we have highlighted the 
  parts that have changed from the simpler version shown in Figures~\ref{fig:multiplicative_alg} and~\ref{fig:oracle_simple}.

\IncMargin{0.5em}
\begin{algorithm}
\DontPrintSemicolon
\SetAlFnt{\small \sf}
\SetKwComment{tcc}{$(*$ }{ $*)$}
\AlFnt
\SetKw{Return}{return}
\SetKwInOut{Input}{Input}\SetKwInOut{Output}{Output}
\Input{A graph $G=(V,E)$ with capacities $\{u_{e}\}_{e}$, a target flow value $F$, edge weights  $\{w_e\}_e$, \newcode{and a set $H$ of forbidden edges}}
\Output{Either a flow $\of$ and a set $H$ of forbidden edges, or $\fail$ indicating that $F>\maxflow$}
\BlankLine 
 \newcode{$\rhop\leftarrow \frac{8m^{1/3}\ln^{1/3} m}{\eps}$}\;
$r_{e}\leftarrow \frac{1}{u_e^2}\left(w_{e}+\frac{\eps |\ww|_{1}}{3m} \right)$ for each $e\in E\newcode{$\setminus H$}$\;

Find an approximate electrical flow $\tf$ using 
  Theorem \ref{thm:computing_electrical_flow} on 
  \newcode{$G_{H}:=(V,E\setminus H)$} with resistances $\rr$,\\
  \quad target flow value $F$,
  and parameter $\delta = \epsilon /3$. \;

{\lIf{\AlFnt{$\energy{\rr}{\tf}>(1+\eps)|\ww|_{1}$ or \newcode{$s$ and $t$ are disconnected in $G_{H}$}}}{\AlFnt{\Return{$\fail$}}}}\;
\newcode{\lIf{\AlFnt{there exists $e$ with $\cong{\tf}{e}>\rhop$}}{\AlFnt{add $e$ to $H$ and start over }}}\;
\Return $\tf$
\caption{The modified oracle $\oracle'$ used by our improved algorithm}\label{fig:improved_oracle}
\end{algorithm}\DecMargin{0.5em}

\IncMargin{0.5em}
\begin{algorithm}
\DontPrintSemicolon
\SetAlFnt{\small \sf}
\SetKwComment{tcc}{$(*$ }{ $*)$}
\AlFnt
\SetKw{Return}{return}
\SetKwInOut{Input}{Input}\SetKwInOut{Output}{Output}
\Input{A graph $G=(V,E)$ with capacities $\{u_{e}\}_{e}$, and a target flow value $F$;}
\Output{Either a flow $\of$, or $\fail$ indicating that $F>\maxflow$;}
\BlankLine
{Initialize $w_{e}^{0}\leftarrow 1$ for all edges $e$, \newcode{$H\leftarrow \emptyset$}, \newcode{$\rhop\leftarrow \frac{8m^{1/3}\ln^{1/3} m}{\eps}$}, and \newcode{$N\leftarrow \frac{2\rhop\ln m }{\eps^{2}}$}}\;
\For{$i:=1,\ldots,N$}{

Query $\oracle'$ with edge weights given by $\ww^{i-1}$, target flow value $F$, \newcode{and forbidden edge set $H$}\;
\lIf{\AlFnt{$\oracle$ returns $\fail$}}{\AlFnt{\Return $\fail$}}\;
\Else{
Let $f^{i}$ be the returned answer\:
\newcode{Replace $H$ with the returned (augmented) set of forbidden edges}\;
$w_{e}^{i}\leftarrow w_{e}^{i-1}(1+\frac{\eps}{\rho}\cong{f^{i}}{e})$ for each $e\in E$\;
}
}
\Return $\of\leftarrow \frac{(1-\eps)^2}{(1+\eps)N}(\sum_{i} f^{i})$\;
\caption{An improved $(1-O(\eps))$-approximation algorithm for the maximum flow problem}\label{fig:improved_alg}
\end{algorithm}\DecMargin{0.5em}

\subsection{Analysis of the New Algorithm}
Before proceeding to a formal analysis of the new algorithm, it will be helpful to examine what is already guaranteed by 
  the analysis from Section~\ref{sec:simple}, and what we'll need to show to demonstrate the algorithm's correctness and bound its 
  running time.  

We first note that, by construction, the congestion of any edge in the flow $\tf$ returned by the modified oracle from 
  Figure~\ref{fig:improved_oracle}  will be bounded by $\rhop$.
Furthermore, it enforces the bound $\energy{\rr}{\tf}\leq (1+\eps)|\ww|_{1}$; by 
  Equations~\eqref{eq:avg_bound}, \eqref{eq:avg_bound2}, and~\eqref{eq:avg_bound3} in Section~\ref{sec:sqrtn_oracle}, this guarantees that $\tf$ will meet the weighted average congestion bound required 
  for a $(\epsilon,\rhop)$-oracle.  
So, as long as the modified oracle always successfully returns a flow, it will function as an $(\epsilon,\rhop)$-oracle,
  and our analysis from Section~\ref{sec:simple} will show that the multiplicative update scheme employed by our algorithm will yield
  an approximate maximum flow after $\Ot{\rhop}$ iterations.

Our problem is thus reduced to understanding the behavior of the modified oracle.  
To prove correctness, we will need to show that whenever the modified oracle is called with $F\leq \maxflow$, it will 
  return some flow $\tf$ (as opposed to returning $\fail$).  
To bound the running time, we will need to provide an upper bound on the total number of electrical flows computed 
  by the modified oracle throughout the execution of the algorithm.

To this end, we will show the following upper bound on the cardinality  $|H|$ and the capacity $u(H)$ of the set of 
  forbidden edges, whose proof we postpone until the next section:
\begin{lemma}\label{lem:cardinality_capacity_bound}
Throughout the execution of the algorithm, $$|H|\leq\frac{30 m\ln m}{\eps^{2}\rhop^{2}}$$ and $$u(H)\leq \frac{30 m F \ln m}{\eps^{2}\rhop^{3}}.$$
\end{lemma}

If we plug in the value $\rhop=(8m^{1/3}\ln^{1/3} m)/\epsilon$ used by the algorithm, Lemma~\ref{lem:cardinality_capacity_bound} gives the bounds
  $|H|\leq \frac{15}{32} (m \ln m)^{1/3}$ and $u(H)\leq \ \frac{15}{256}\epsilon F< \epsilon F/12$.

Given the above lemma, it is now straightforward to show the following theorem, which establishes the correctness 
  and bounds the running time of our algorithm.

\begin{theorem}
For any $0<\epsilon<1/2$,
  if $F\leq \maxflow$ the algorithm in Figure \ref{fig:improved_alg} will return  a feasible $s$-$t$ flow $\of$ of value $\thr{\of}=(1-O(\eps))F$ 
  in time $\Ot{m^{4/3}\eps^{-3}}$.
\end{theorem}

\begin{proof}
To bound the running time, we note that, whenever we invoke the algorithm from 
  Theorem \ref{thm:computing_electrical_flow} , we either advance the number of iterations or we increase the cardinality 
  of $H$, so the number of linear systems we solve is at most $N+|H|\leq N +  \frac{15}{32} (m \ln m)^{1/3}$.

Equation \eqref{eq:R_bound} implies that the value of $R$ from Theorem \ref{thm:computing_electrical_flow} is $O((m/\eps)^{O(1)})$, so
  solving each linear system takes time at most $\Ot{m\log 1/\eps}$.  
This gives an overall running time of  
$$\Ot{\left(N+\textstyle{\frac{15}{32}}(m \ln m)^{1/3}\right)m}=\Ot{m^{4/3}\eps^{-3}},$$
  as claimed.

It thus remains to prove correctness.  
For this, we need to show that if $F\leq \maxflow$, then the oracle does not return $\fail$, which would occur if we 
  disconnect $s$ from $t$ or if $\energy{\rr}{\tf}>(1+\epsilon)|\ww|_{1}$.  
By Lemma~\ref{lem:cardinality_capacity_bound} and the comment following it, we know that throughout the whole algorithm $G_{H}$ has maximum 
  flow value of at least 
  $\maxflow - \epsilon F/12 \geq \left(1-\epsilon/12\right)F $
  and thus, in particular, we will never disconnect $s$ from $t$. 

Furthermore, this implies that there exists a feasible flow in our graph of value $\left(1-\epsilon/12\right)F$, even after we have removed the edges in $H$.  There is thus a flow of value $F$ in which every edge has congestion at most $1/\left(1-\epsilon/12\right)$.  
We can therefore use the argument from Section~\ref{sec:sqrtn_oracle} (Equation~\eqref{eq:approx_flow_energy_bound} and the lines directly preceding it) to show that we always have 
$$\energy{\rr}{\tf}\leq (1+\eps/12)^2(1+\eps/3)^2|\ww|_{1} \leq (1+\eps)|\ww|_{1}, $$
  as required. 
\end{proof}

The above theorem allows us to apply the binary search strategy that we used in Section~\ref{sec:MW_overview}.
This yields our main theorem:

\begin{theorem}\label{thm:result_improved}
For any $0 <\epsilon <1/2$,
   the maximum flow problem can be $(1-\eps)$-approximated in $\Ot{m^{4/3}\eps^{-3}}$ time.
\end{theorem}

\subsection{The Proof of Lemma~\ref{lem:cardinality_capacity_bound}}

All that remains is to prove the bounds given by Lemma~\ref{lem:cardinality_capacity_bound} on the cardinality and capacity of $H$.  
To do so, we will use the effective resistance of the circuits on which we compute electrical flows as a potential function.  
The key insight is that we only cut an edge when its flow accounts for a nontrivial fraction of the energy of the electrical flow, 
  and that cutting such an edge will cause a substantial change in the effective resistance.   
Combining this with a bound on how much the effective resistance can change during the execution of the algorithm will 
  guarantee that we won't cut too many edges.

Let $\rrt{j}$ be the resistances used in the $j^\text{th}$ electrical flow computed during the execution of the algorithm%
\footnote{Note that $\rrt{j}$ is not just the set of  resistances arising from $\ww^{j}$, since a single call to the oracle may 
  compute multiple electrical flows as edges  are added to $H$.}.   
If an edge $e$ is not in $E$ or if $e$ has been added to $H$ by step $j$, set $\rt{j}{e}=\infty$.  
We define the potential function
$$\Phi(j)=\Reff{\rrt{j}}= \energy{\rrt{j}}{\flowt{j}},$$
where $\flowt{j}$ is the (exact) electrical flow of value 1 arising from $\rrt{j}$.
Lemma~\ref{lem:cardinality_capacity_bound} will follow easily from:
\begin{lemma}\label{lem:phi_props}
Suppose that $F \leq \maxflow \leq mF$.  Then:
\begin{enumerate}
\item $\Phi(j)$ never decreases during the execution of the algorithm. \label{lem:phi_nondec}
\item $\Phi(1)\geq m^{-4}F^{-2}$. \label{lem:phi_lb}
\item If we add an edge to $H$ between steps $j-1$ and $j$, then $(1-\frac{\eps\rhop^{2}}{5m})\Phi(j)> \Phi(j-1)$. \label{lem:phi_change}
\end{enumerate}
\end{lemma}
\begin{proof}
\subsubsection*{Proof of (1)} 
The only way that the resistance $\rt{j}{e}$ of an edge $e$ can change is if the weight $w_e$ is increased by the 
  multiplicative weights routine, or if $e$ is added to $H$ so that $\rt{j}{e}$ is set to $\infty$.  
As such, the resistances are nondecreasing during the execution of the algorithm.  
By Rayleigh Monotonicity (Corollary~\ref{cor:rayleigh_mon}), this implies that the effective resistance is nondecreasing as well.

\subsubsection*{Proof of (2)}
In the first linear system, $H=\emptyset$ and $\rt{1}{e}=\frac{1+\eps/3}{u_{e}^{2}}$ for all $e\in E$.
Let $(S,V\setminus S)$ be the minimum $s$-$t$ cut of $G$. 
By the Max Flow-Min Cut Theorem (\cite{FordF56,EliasFS56}), we know that the capacity $u(S)=\sum_{e\in E(S)} u_{e}$ of this cut is equal to $\maxflow$.   
In particular,
$$\rt{1}{e}=\frac{1+\eps/3}{u_{e}^{2}}\geq \frac{1+\eps/3}{{\maxflow}^2}> \frac{1}{{\maxflow}^2} $$
for all $e\in E(S)$.
As $\flowt{1}$ is an electrical $s$-$t$ flow of value 1,
 it sends 1 unit of net flow across $(S,V\setminus S)$; so,
  some edge $e'\in E(S)$ must have $\flowt{1} (e')\geq 1/m$.   
This gives
\begin{equation}\label{eq:eff_res_lb}
\Phi(1)=\energy{\rrt{1}}{\flowt{1}}
=\sum_{e\in E}\flowt{1}(e)^2\rt{1}{e} 
\geq \flowt{1}(e')^2\rt{1}{e'} 
> \frac{1}{m^2{\maxflow}^2}.
\end{equation}

Since $\maxflow \leq mF$ by assumption, the desired inequality follows.

\subsubsection*{Proof of (3)}
Suppose we add the edge $h$ to $H$ between steps $j-1$ and $j$.  We will show that $h$ accounts for a substantial fraction of the total energy of the electrical flow with respect to the resistances $\rrt{j-1}$, and our result will then follow from Lemma~\ref{lem:effect_of_res_incr}.

Let $\ww$ be the weights used at step $j-1$, and let $\tf$ be the flow we computed in step $j-1$.
Because we added $h$ to $H$, we know that $ \cong{\tf}{h}> \rhop$.  
Since our algorithm did not return $\fail$ after computing this $\tf$, we must have that 
\begin{equation}\label{eq:energy_bound_lem_phi}
\energy{\rrt{j-1}}{\tf}\leq (1+\eps)|\ww|_{1}.
\end{equation}
Using the definition of $\rt{j-1}{h}$, the fact that $ \cong{\tf}{h}>\rhop$, and  Equation~\eqref{eq:energy_bound_lem_phi}, we obtain:
\begin{align*}
{\tf(h)^2\rt{j-1}{h}}&= 
\tf(h)^2\frac{w_e+\epsilon \frac{|\ww|_1}{3m}}{u_e^2}\\
&\geq \tf(h)^2\frac{\epsilon {|\ww|_1}}{3mu_e^2}\\
&=\frac{\epsilon }{3m} \left(\frac{\tf(h)}{u_e}\right)^2 {|\ww|_1}\\
&=\frac{\epsilon }{3(1+\epsilon)m} \cong{\tf}{h}^2 \left((1+\epsilon){|\ww|_1}\right)\\
&> \frac{\epsilon \rhop^2}{3(1+\epsilon)m} \energy{\rrt{j-1}}{\tf}.
\end{align*}

The above inequalities establish that  edge $h$ accounts for more than a 
  $\frac{\eps\rhop^{2}}{3(1+\eps)m}$ fraction of the total energy  $\energy{\rr^{i}}{\tf}$ of the flow $\tf$.

The flow $\tf$ is the approximate electrical flow computed by our algorithm, but our argument will require that an inequality
  like this holds in the exact electrical flow $\flowt{j-1}$.  
This is guaranteed by  part $b$ of Theorem~\ref{thm:computing_electrical_flow}, which,
along with the facts that 
 $\energy{\rr}{\tf} \geq  \energy{\rr}{\flowt{j-1}}$, $\rho \leq 1$,
   and $\epsilon < 1/2$,
gives us that 
$$\flowt{j-1}(h)^2\rt{j-1}{h}
>{\tf(h)^2\rt{j-1}{h}}- \frac{\epsilon/3}{2mR}\energy{\rr}{\flowt{j-1}}
>\left(\frac{\epsilon \rhop^2}{3(1+\epsilon)m}- \frac{\epsilon/3}{2mR}\right)\energy{\rr}{\flowt{j-1}}>\frac{\epsilon \rhop^2}{5m}\energy{\rr}{\flowt{j-1}}.
$$
The result now follows from Lemma~\ref{lem:effect_of_res_incr}.
\end{proof}

We are now ready to prove Lemma~\ref{lem:cardinality_capacity_bound}.

\begin{proof}[Proof of Lemma~\ref{lem:cardinality_capacity_bound}]
Let $k$ be the cardinality of the set $H$ at the end of the algorithm. 
Let $\tf$ be the flow that was produced by our algorithm just before $k$-th edge was added to $H$, let $j$ be the time 
  when this flow was output, and let $\ww$ be the corresponding weights. 

As the energy required by an $s$-$t$ flow scales with the square of the value
  of the flow,
\begin{equation}\label{eq:phi_by_energy}
\Phi(j) 
  \leq \frac{\energy{\rr^{j}}{f}}{F^{2}}
  \leq \frac{\energy{\rr^{j}}{\tf}}{F^{2}}.
\end{equation} 
By the construction of our algorithm, it must have been the case that $ \energy{\rrt{j}}{\tf} \leq (1+\eps)|\ww|_{1}$. 
This inequality together with equation \eqref{eq:phi_by_energy} and part~\ref{lem:phi_lb} of Lemma~\ref{lem:phi_props} implies that
\[
\Phi(j)=\energy{\rrt{j}}{\flowt{j}}\leq \frac{\energy{\rrt{j}}{\tf}}{F^{2}}\leq (1+\eps)|\ww|_{1}m^{4}\Phi(1).
\] 

Now, since up to time $j$ we had $k-1$ additions of edges to $H$, parts~\ref{lem:phi_nondec} and~\ref{lem:phi_change} of Lemma \ref{lem:phi_props}, and  Lemma \ref{lem:upper_sum_potential}  imply that

\[
\left(1-\frac{\eps\rhop^{2}}{5m}\right)^{-(k-1)}\leq \frac{\Phi(j)}{\Phi(1)}
\leq(1+\eps)|\ww|_{1}m^{4}
\leq(1+\eps)m^{4} \left( m \exp\left(\textstyle{\frac{(1+\epsilon)\epsilon}{\rho}}N\right) \right)
\leq 2m^{5}\exp(3\eps^{-1}\ln m),
\]
where the last inequality used the fact that $\eps<1/2$.
Rearranging the terms in the above inequality gives us that

$$k\leq - \frac{\ln 2+5\ln m +3\epsilon^{-1}\ln m}{ \ln \left(1-\frac{\eps\rhop^{2}}{5m}\right)} +1
< - \frac{6\epsilon^{-1}\ln m}{ \ln \left(1-\frac{\eps\rhop^{2}}{5m}\right)}  
<  \frac{30 m\ln m}{ \epsilon^2\rho^2},$$
where we used the inequalities $\epsilon<1/2$ and $\log(1-c)<-c$ for all $c\in (0,1)$.
This establishes our bound on cardinality of the set $H$. 

To bound the value of $u(H)$, let us note that we add an edge $e$ to $H$ only  when we send at least $\rhop u_{e}$ units of flow across it. 
But since we never flow more than $F$ units across any single edge, we have that $u_{e}\leq F/\rhop$. 
Therefore, we may conclude that
\[
u(H)\leq |H|\frac{F}{\rhop}\leq \frac{30m F \ln m}{\eps^{2}\rhop^{3}},
\] 
as desired.
\end{proof}



\subsection{Improving the Running Time to \texorpdfstring{$\Ot{mn^{1/3}\epsilon^{-11/3}}$}{\~O(mn\textasciicircum(1/3) eps\textasciicircum(-11/3))}}

We can now combine our algorithm with existing methods to further improve its running time.
In \cite{Karger98} (see also \cite{BenczurK02}), Karger presented a technique, which he called  ``graph smoothing'', 
that allows one to 
use random sampling to speed up an exact or $(1-\epsilon)$-approximate flow algorithm.
More precisely, his techniques yield the following theorem, which is implicit in \cite{Karger98} and stated in a more similar form in~\cite{BenczurK02}:
\begin{theorem}[\cite{Karger98,BenczurK02}]\label{thm:smoothing}
Let $T(m,n,\eps)$ be the time needed to find a $(1-\eps)$-approximately maximum flow in an undirected, capacitated graph with $m$ edges and $n$ vertices. Then one can obtain a $(1-\eps)$-approximately maximal flow in such a graph in time $\Ot{\eps^2m/n\cdot T(\Ot{n\eps^{-2}},n,\Omega(\eps))}$.
\end{theorem}
By applying the above theorem to our $\Ot{m^{4/3}\eps^{-3}}$ algorithm, we obtain our desired running time bound:
\begin{theorem}
For any $0<\epsilon<1/2$, the maximum flow problem can be $(1-\eps)$-approximated in $\Ot{mn^{1/3}\eps^{-11/3}}$ time.
\end{theorem}

\subsection{Approximating the Value of the Maximum $s$-$t$ Flow in Time \texorpdfstring{$\Ot{m+n^{4/3}\epsilon^{-8/3}}$}{\~O(m+n\textasciicircum(4/3) eps\textasciicircum(-8/3))}}\label{sec:approx_value}

 Given any weighted undirected graph $G=(V,E,w)$ with $n$ vertices and $m$ edges, Bencz{\'u}r and Karger~\cite{BenczurK96}
  showed that one can construct a graph $G'=(V,E',w')$ (called a \emph{sparsifier of $G$}) on the same vertex set in 
  time $\Ot{m}$ such that $|E'|=O(n \log n/\epsilon^2)$ and the capacity of any cut in $G'$ is between $1$ and $(1+\epsilon)$ 
  times its capacity in $G$. 
Applying our algorithm from Section~\ref{sec:improved_flow} to a sparsifier of $G$ gives us an algorithm for $(1-\epsilon)$-approximating
  the value of the maximum $s$-$t$ flow on $G$ in time $\Ot{m+n^{4/3}\epsilon^{-3}}$.

We note that this only allows us to approximate the \emph{value} of the maximum $s$-$t$ flow on $G$.  
It gives us a flow on $G'$, not one on $G$.  
We do not know how to use an approximately maximum $s$-$t$ flow on $G'$ to obtain one on $G$ in less time than would be required to 
  compute a maximum flow in $G$ from scratch using the algorithm from Section~\ref{sec:improved_flow}.

For this reason, there is a gap between the time we require to find a maximum flow and the time we require to compute its value.  We note, however,
 that this gap will not exist for the minimum $s$-$t$ cut problem, since an approximately minimum $s$-$t$ cut on $G'$ will also be an
  approximately minimum $s$-$t$ cut $G$.
We will present an algorithm for finding such a cut in the next section.
By the Max Flow-Min Cut Theorem, this will provide us with an alternate algorithm for approximating the value of the maximum $s$-$t$ flow.
It will have a slightly better dependence on $\epsilon$, which will allow us to approximate the value of the maximum $s$-$t$ flow
  in time $\Ot{m+n^{4/3}\epsilon^{-8/3}}$.


\section{A Dual Algorithm for Finding an Approximately Minimum $s$-$t$ Cut in Time \texorpdfstring{$\Ot{m+n^{4/3}\epsilon^{-8/3}} $}{\~O(m+n\textasciicircum(4/3) eps\textasciicircum(-8/3))}}\label{sec:dual_alg}
\newcommand{\rhopc}{\rho}

In this section, we'll describe a dual perspective that	 yields to an even simpler algorithm for computing an approximately
  minimum $s$-$t$ cut.  
Rather than using electrical flows to obtain a flow, it will use the electrical potentials to obtain a cut.  

The algorithm will eschew the oracle abstraction and multiplicative weights machinery.  
Instead, it will just repeatedly compute an electrical flow, increase the resistances of edges
  according to the amount flowing over them, and repeat.  
It will then use the electrical potentials of the last flow computed
to  find a cut by picking a  
 cutoff and splitting the vertices according to whether their potentials are above
  or below the cutoff. 

The algorithm is shown in Figure~\ref{fig:cut_alg}.  
It finds a
  $(1+\epsilon )$-approximately minimum $s$-$t$ cut in time
  $\Ot{m^{4/3} \epsilon^{-8/3}}$; 
  applying it to a sparsifier 
  will give us:
\begin{theorem}\label{thm:sparsified_cut}
  For any $0<\eps<1/7$, we can find a $(1+\eps)$-approximate minimum $s$-$t$ cut in $\Ot{m+n^{4/3}\eps^{-8/3}}$ time.
\end{theorem}

We note that, in this algorithm, there is no need to deal explicitly with edges flowing more than $\rhopc$, maintain 
  a set of forbidden edges, or average the flows from different steps.  
We will separately study edges with very large flow in our analysis, but the algorithm itself 
  avoids the complexities that appeared in the improved flow algorithm described in 
  Section~\ref{sec:improved_flow}.

We further note that the update rule is slightly modified from the one that appeared 
  earlier in the paper.  
This is done to guarantee that the effective resistance increases substantially when 
  some edge flows more than $\rhopc$, without having to explicitly cut it.  
Our previous rule allowed the weight (but not resistance) of an edge to constitute a 
  very small fraction of the total weight; 
  in this case, a significant multiplicative increase in the weight of an edge may not
  produce a substantial change in the effective resistance of the graph.

\begin{algorithm}
\DontPrintSemicolon
\SetAlFnt{\small \sf}
\AlFnt
\SetKw{Return}{return}
 \SetKwInOut{Input}{Input}\SetKwInOut{Output}{Output}
\Input{A graph $G=(V,E)$ with capacities $\{u_{e}\}_{e}$, and a target flow value $F$}
\Output{A cut $(S,V\setminus S)$}
\BlankLine
{Initialize $w_{e}^{0}\leftarrow 1$ for all edges $e$, $\rhopc \leftarrow 3m^{1/3} \epsilon^{-2/3}$, $N\leftarrow 
 5\epsilon^{-8/3} m^{1/3} \ln m $}, and $\delta \leftarrow \epsilon^{2}$.\;
\For{$i:=1,\ldots,N$}{
Find an approximate electrical flow $\tf^{i-1}$ 
  and potentials $\tphiphi$ 
  using Theorem \ref{thm:computing_electrical_flow} on $G$ with \;
 \quad resistances $r_e^{i-1}=\frac{w_e^{i-1}}{u_e^2}$, target flow value $F$, and parameter $\delta $. \;

$\mu^{i-1}\leftarrow \sum_e{w_e^{i-1}} $\;
$w_e^{i}\leftarrow w_e^{i-1} + \frac {\epsilon}{\rhopc} \cong{\tf^{i-1}}{e} w^{i-1}_e + \frac {\epsilon^2}{m \rhopc} \mu^{i-1}$ for each $e\in E$\;
Scale and translate $\tphiphi$ so that $\tphi_s=1$ and $\tphi_t=0$\;
Let $S_{x} = \{v \in V \,|\, \phi_v>x\}$\;
Set $S$ to be the set $S_{x}$ that minimizes $(S_{x}, V \setminus S_{x})$\;
If the capacity of $(S_{x}, V \setminus S_{x})$ is less than $F / (1-7 \epsilon)$, \Return $(S_{x}, V \setminus S_{x})$.
}
\Return \fail``\;

\caption{A dual algorithm for finding an $s$-$t$ cut}\label{fig:cut_alg}
\end{algorithm}

\subsection{An Overview of the Analysis}

To analyze this algorithm, we will track the total weight placed on the edges 
  crossing some minimum cut.  
The basic observation for our analysis is that the same amount of net flow 
  must be sent across every cut, so edges in small cuts will tend to have 
  higher congestion than edges in large cuts.
Since our algorithm increases the weight of an edge according to its congestion, 
  this will cause our algorithm to concentrate a larger and larger fraction of the 
  total weight on the small cuts of the graph.   
This will continue until almost all of the weight is concentrated on 
  approximately minimum cuts.

Of course, the edges  crossing a minimum cut will also cross many other 
  (likely much larger) cuts, so we certainly can't hope to obtain a graph in 
  which no edge crossing a large cut has non-negligible weight.  
In order to formalize the above argument, we will thus need some good 
  way to measure the extent to which the weight is ``concentrated on 
  approximately minimum cuts''.  

In Section~\ref{sec:cuts_and_pots}, we will show how to use effective resistance
  to formulate such a notion.   
In particular,
  we will show that if we can make the effective resistance large
  enough then we can find a cut of small capacity.
In Section~\ref{sec:cut_alg_analysis}, we will use an argument like the one 
  described above to show that the resistances produced by the algorithm in 
  Figure~\ref{fig:cut_alg}
  converge after $N=\Ot{m^{1/3} \epsilon^{-8/3}}$ steps to one that meets such a bound.

\subsection{Cuts, Electrical Potentials, and Effective Resistance }
\label{sec:cuts_and_pots} 
During the algorithm, we scale and translate the
  potentials of the approximate  electrical flow so that $\tphi_s=1$
  and $\tphi_t=0$.  
We then produce a cut by choosing $x\in [0,1]$
  and dividing the graph into the sets 
  $S=\{ v \in V\, | \, \phi_v > x\}$ 
  and $V\setminus S=\{ v \in V \, |\, \tphi_v \leq x\}$.  
The following lemma upper bounds the capacity of the
  resulting cut in terms of the electrical potentials and edge
  capacities.

\begin{lemma}\label{lem:edges_cut}
Let $\tphiphi$ be as above.  
Then there is a cut $S_{x}$ of capacity
  at most
\begin{equation}\label{eqn:edges_cut}
\sum_{(u,v)\in E} |\tphi_u - \tphi_v|u_{(u,v)}.
\end{equation}
\end{lemma}
\begin{proof}
Consider choosing $x \in [0,1]$ uniformly at random.
The probability that an edge $(u,v)$ is cut is precisely $|\tphi(u) - \tphi(v)|$.
So, the expected capacity of the edges in a random cut is given by \eqref{eqn:edges_cut},
  and so there is a cut of capacity at most \eqref{eqn:edges_cut}.
\end{proof}

Now, suppose that one has a fixed total amount of resistance $\mu$ to distribute 
  over the edges of a cut of size $F$.  
It is not difficult to see that the maximum possible effective resistance between $s$
  and $t$ in such a 
  case is $\frac{\mu}{F^2}$, and that this is achieved when  one puts a resistance 
  of  $\frac{\mu}{F}$ on each of the edges.  
This suggests the following lemma, which bounds the quantity in 
  Lemma~\ref{lem:edges_cut} in terms of the effective resistance and the total 
  resistance (appropriately weighted when the edges have non-unit capacities):  

\begin{lemma}\label{lem:potdrop_and_reff}
Let $\mu = \sum_e {u_e^2 r_e}$, and let the effective $s$-$t$ resistance 
  of $G$ with edge resistances given by $\rr$ be $\Reff{\rr}$.
Let $\phiphi$ be the potentials of the electrical $s$-$t$ flow, scaled
  to have potential drop $1$ between $s$ and $t$.
Then
\[
 \sum_{e \in E}{\phiphi(e) u_e} \leq \sqrt{\frac{\mu }{\Reff{\rr}}} .
\]
If, $\tphiphi$ is an approximate electrical potential returned by 
  the algorithm of Theorem~\ref{thm:computing_electrical_flow}
  when run with parameter $\delta \leq 1/3 $, re-scaled to have
  potential difference $1$ between $s$ and $t$, then
\[
 \sum_{e \in E}{\tphiphi(e) u_e} \leq (1 + 2 \delta ) \sqrt{\frac{\mu }{\Reff{\rr}}} .
\]
\end{lemma}
\begin{proof}
By Fact~\ref{fa:effective_cond_express},
  the rescaled true electrical potentials correspond to a flow of value
  $1/\Reff{\rr}$ and
\[
  \sum_{e} \frac{\phiphi (e)^{2}}{r_{e}} = \frac{1}{ \Reff{\rr}}.
\]
So, we can apply the Cauchy-Schwarz inequality to prove
\begin{align*}
\sum_e{\phi(e) u_e} 
&\leq \sqrt{\sum_e{\frac {\phi(e)^2}{r_e}} \sum_e{u_e^2 r_e}} \\
& =  \sqrt{\frac{\mu }{\Reff{\rr}}}.
\end{align*}
By parts $a$ and $c$ of Theorem~\ref{thm:computing_electrical_flow}, after we rescale
  $\tphiphi$ to have potential drop $1$ between $s$ and $t$, it will have energy
\[
  \sum_{e} \frac{\tphiphi (e)^{2}}{r_{e}} 
\leq \frac{1+\delta}{1 - \delta }\frac{1}{ \Reff{\rr}}
\leq  (1 + 3 \delta )\frac{1}{ \Reff{\rr}},
\]
for $\delta \leq 1/3$.
The rest of the analysis follows from another application of Cauchy-Schwarz.
\end{proof}

\subsection{The Proof that the Dual Algorithm Finds an Approximately Minimum Cut} \label{sec:cut_alg_analysis}

We'll show that
  if $\maxflowguess \geq \maxflow$ then
 within $N=5 \epsilon^{-8/3} m^{1/3} \ln m$ iterations, the algorithm in 
  Figure~\ref{fig:cut_alg} will produce a set of resistances $\rr^i$ such that 
\begin{equation}\label{eqn:ReffTarget}
\Reff{\rr^i} \geq (1 - 7 \epsilon) \frac {\mu^i}{\left(\maxflowguess\right)^2}.
\end{equation}
Once such a set of resistances has been obtained, Lemmas~\ref{lem:edges_cut} and 
  \ref{lem:potdrop_and_reff} tell us that the best potential cut of $\tphiphi$
  will have capacity at most
\[
  \frac{1 + 2 \delta }{\sqrt{1 - 7 \epsilon}} F
 \leq 
  \frac{1}{1 - 7 \epsilon} F.
\]
The algorithm will then return this cut.

Let $C$ be the set of edges crossing some minimum cut in our graph.
Let $u_{C}= \maxflow$ denote the capacity of the edges in $C$.
We will keep track of two quantities: the weighted geometric mean
  of the weights of the edges in $C$,
\[
  \nu^i = \left(
              \prod_{e \in C} \left(w_e^i \right)^{u_{e}}
         \right)^{1/u_{C}},
\]
 and the total weight
\[
  \mu^{i}= \sum_e{w_e^{i}} = \sum_e r_e^i u_e^2 
\]
 of the edges of the entire graph.
Clearly $\nu^i \leq \max_{e \in C}{w_e^{i}}$.
In particular, 
\[ \nu^i \leq \mu^i \]
for all $i$.

Our proof that the effective resistance cannot remain large for too many iterations 
  will be similar to our analysis of the flow algorithm in Section~\ref{sec:improved_flow}.  
We suppose to the contrary that 
  $\reff^i \leq (1 - 7 \epsilon) \frac {\mu^i}{\left(\maxflowguess\right)^2}$ for each $1 \leq i \leq N$.  
We will show that, under this assumption:
\begin{enumerate}
\item The total weight $\mu^i$ doesn't get too large over the course of the algorithm  {\bf [Lemma~\ref{lem:total_weight}]}.
\item \label{type1} The quantity $\nu^i$ increases significantly in any iteration in which no edge 
  has congestion more than $\rhopc$ {\bf [Lemma~\ref{lem:nu_increase}]}.  Since $\mu^i$ doesn't get too large, 
  and $\nu^i \leq \mu^i$, this will not happen too many times.
\item \label{type2} The effective resistance increases significantly in any iteration 
  in which some edge has congestion more than $\rhopc$ {\bf [Lemma~\ref{lem:eff_res_increase}]}.  
Since $\mu^i$ does not get too large, and the effective resistance is assumed 
  to be bounded in terms of the total weight $\mu^i$, this cannot happen 
  too many times.
\end{enumerate}
The combined bounds from~\ref{type1}  and~\ref{type2} will be less than $N$, 
  which will yield a contradiction.

\begin{lemma}\label{lem:total_weight}
For each $i \leq N$ such that 
  $\Reff{\rr^{i}} \leq (1 - 7 \epsilon) \frac{\mu^{i}}{\maxflowguess^{2}}$,
\[ \mu^{i+1} \leq \mu^i \exp\left(\frac {\epsilon(1 - 2\epsilon)}{\rhopc} \right). \]
So, if for all $i \leq N$ we have 
  $\Reff{\rr^{i}} \leq (1 - 7 \epsilon) \frac{\mu^{i}}{\maxflowguess^{2}}$, 
then
\begin{equation}
\label{eq:mu_bound}
\mu^N \leq \mu^0 \exp\left(\frac {\epsilon(1 - 2\epsilon)}{\rhopc} N\right).
\end{equation}
\end{lemma}
\begin{proof}
If $\Reff{\rr^i} \leq (1 - 7 \epsilon) \frac {\mu^i}{F^2} $
  then the electrical flow $f$ of value $F$ has energy
\[
  F^{2} \Reff{\rr^{i}}
 = 
  \sum{\cong{f^i}{e}^2 w_e^i}
 \leq (1 - 7 \epsilon) 
   \mu^i.
\]
By Theorem~\ref{thm:computing_electrical_flow}, the approximate electrical flow $\tf$
  has energy at most $(1+\delta)$ times the 
  energy of $f$,
\[
 \sum{\cong{\tf^i}{e}^2 w_e^i}
 \leq (1+\delta ) (1 - 7 \epsilon) 
   \mu^i
 \leq (1 - 6 \epsilon) 
   \mu^i.
\]
Applying the Cauchy-Schwarz inequality,
  we find
\[
 \sum{\cong{\tf^i}{e} w_e^i}
\leq 
  \sqrt{\sum w_{e}^{i}} 
  \sqrt{\sum \cong{\tf^i}{e}^2 w_{e}^{i}}
\leq 
  \sqrt{1 - 6 \epsilon} \mu^{i}
\leq 
  (1 - 3 \epsilon) \mu^{i}.
\]
Now we can compute
\begin{align*}
\mu^{i+1}
&= \sum_e{w_e^{i+1}} \\
&= \sum_e{w_e^i} + \frac {\epsilon}{\rhopc} \sum_e{w_e^i \cong{f^i}{e}} + \frac {\epsilon^2}{\rhopc} \mu^i \\
&\leq \left( 1  + \frac {\epsilon(1 - 2 \epsilon)}{\rhopc}\right) \mu^i \\
& \leq \mu^i \exp\left(\frac {\epsilon(1 - 2\epsilon)}{\rhopc} \right)
\end{align*}
as desired.
\end{proof}

\begin{lemma}\label{lem:nu_increase}
If $\cong{f^i}{e} \leq \rhopc$ for all $e$, then
\[
  \nu^{i+1} \geq 
  \exp\left( \frac {\epsilon(1 - \epsilon)}{\rhopc}\right) \nu^i  .
\]
\end{lemma}
\begin{proof}
To bound the increase of $\nu^{i+1}$ over $\nu^{i}$, we use the inequality
\[
  (1 + \epsilon x) \geq \exp \left(\epsilon (1-\epsilon) x \right),
\]
which holds for $\epsilon$ and $x$ between $0$ and $1$.
We apply this inequality with $x = \cong{\tf^i}{e} / \rho$.
As $\tf^{i}$ is a flow of value $F$ and $C$ is a cut,
 $\sum_{e \in C}{\abs{\tf_e^i}} \geq  \maxflowguess $.
We now compute
\begin{align*}
\nu^{i+1}
&= 
  \left(\prod_{e \in C}\left( w_e^{t+1} \right)^{u_{e}} \right)^{1/u_{C}}
\\
& \geq  
  \left(
     \prod_{e \in C} \left(w_e^{i}
     \left(1 + \frac{\epsilon}{\rhopc} \cong{\tf^{i}}{e}  \right)\right)^{u_{e}}
   \right)^{1/u_{C}}
\\
& =
  \nu^{i}
  \left(
     \prod_{e \in C} 
     \left(1 + \frac{\epsilon}{\rhopc} \cong{\tf^{i}}{e} \right)^{u_{e}}
   \right)^{1/u_{C}}
\\
& \geq 
  \nu^{i}
  \exp
   \left( \frac{1}{u_{C}}
     \sum_{e \in C} u_{e}
      \frac{\epsilon (1-\epsilon )}{\rhopc} \cong{\tf^{i}}{e}
 \right)
\\
& =
  \nu^{i}
  \exp
   \left( \frac{1}{u_{C}}
     \sum_{e \in C} 
      \frac{\epsilon (1-\epsilon )}{\rhopc} \abs{\tf^{i}_{e}}
 \right)
\\
& \geq 
  \nu^{i}
  \exp
   \left( \frac{1}{u_{C}}
      \frac{\epsilon (1-\epsilon )}{\rhopc} F
 \right)
\\
& \geq 
  \nu^{i}
  \exp
   \left( 
      \frac{\epsilon (1-\epsilon )}{\rhopc} 
 \right)
.
\end{align*}
\end{proof}

\begin{lemma}\label{lem:eff_res_increase}
If
  $\Reff{\rr^{i}} \leq (1 - 7 \epsilon) \frac{\mu^{i}}{\maxflowguess^{2}}$
  and
   there exists some edge $e$ such that 
$\cong{\tf^i}{e}>\rhopc$, 
then 
\[
\Reff{\rr^{i+1}}
  \geq 
\exp\left( \frac {\epsilon^2 \rhopc^2}{4 m} \right) \Reff{\rr^i}.
\]
\end{lemma}
\begin{proof}
We first show by induction that 
$$w_{e}^{i} \geq \frac {\epsilon}{m} \mu^i. $$
If $i = 0$ we have $1 \geq \epsilon$.   For $i > 0$, we have
\begin{align*}
w_e^{i+1} 
&\geq w_e^i + \frac {\epsilon^2}{m \rhopc} \mu^i \\
&\geq \left( \frac {\epsilon}{m} + \frac {\epsilon^2}{m \rhopc}\right) \mu^i \\
& = \frac {\epsilon}{m} \left( 1 + \frac {\epsilon}{\rhopc}\right) \mu^i \\
& \geq  \frac {\epsilon}{m} \exp \left(\frac{\epsilon (1-2 \epsilon)}{\rhopc } \right)
   \mu^i \\
&\geq \frac {\epsilon}{m} \mu^{i+1},
\end{align*}
by Lemma~\ref{lem:total_weight}.

We now show that an edge $e$ for which
  $\cong{\tf^i}{e} \geq \rhopc$ contributes a large
  fraction of the energy to the true electrical flow of value $F$.
By the assumptions of the lemma, the energy of the true electrical
  flow of value $F$ is
\[
  F^{2} \Reff{\rr} \leq (1-7\epsilon) \mu^{i}.
\]
On the other hand, the energy of the edge $e$ in the approximate electrical
  flow is
\[
 \tf^i(e)^2 r_e^i = \cong{\tf^i}{e}^2 w_e \geq \rhopc^2 \frac{\epsilon}{m} \mu^i.
\]
As a fraction of the energy of the true electrical flow, this is at least
\[
  \frac{1}{1-7\epsilon } \frac{\rhopc^{2} \epsilon}{m} 
 =
  \frac{1}{(1-7\epsilon) \epsilon^{1/3} m^{1/3} }.
\]
By part $b$ of Theorem~\ref{thm:computing_electrical_flow}, the fraction of the
  energy that $e$ accounts for in the true flow is at least
\[
  \frac{1}{(1-7\epsilon) \epsilon^{1/3} m^{1/3} } - \frac{\epsilon^{2}}{2 m R}
 \geq 
    \frac{1}{\epsilon^{1/3} m^{1/3} } = 
  \frac{\rhopc^{2} \epsilon}{m}.
\]
As
\[
  w_e^{i+1} 
  \geq w_e^i \left( 1 + \frac {\epsilon}{\rhopc} \cong{\tf^i}{e} \right) 
  \geq (1 + \epsilon) w_e^i,
\]
we have increased the weight of edge $e$ by a factor of at least $(1+\epsilon)$.
So, by Lemma~\ref{lem:effect_of_res_incr},
\[
  \frac{\Reff{\rr^{i+1}}}{\Reff{\rr^{i}}}
\geq
  \left(1 + \frac{\rhopc^{2} \epsilon^{2}}{2 m} \right)
\geq
 \exp  \left(\frac{\rhopc^{2} \epsilon^{2}}{4 m} \right).
\]
\end{proof}

We now combine these lemmas to obtain our main bound:
\begin{lemma}\label{thm:res_conv_bound}
For $\epsilon \leq 1/7$, after $N$ iterations, the algorithm in Figure~\ref{fig:cut_alg} will produce 
  a set of resistances such that
  $\Reff{\rr^i} \leq (1 - 7\epsilon) \frac {\mu^i}{\maxflowguess^2}$.
\end{lemma}

Before proving Lemma~\ref{thm:res_conv_bound}, we note that combining 
  it with Lemmas~\ref{lem:edges_cut} and~\ref{lem:potdrop_and_reff} immediately implies our main result:
\begin{theorem}
On input $\epsilon < 1/7$, the algorithm in Figure~\ref{fig:cut_alg} runs in time
  $\Ot{m^{4/3} \epsilon^{-8/3}}$.
If $F \geq F^{*}$, then it returns a cut of capacity at most $F / (1-7\epsilon)$, where
  $F^{*}$ is the minimum capacity of an $s$-$t$ cut.
\end{theorem}
To use this algorithm to find a cut of approximately minimum capacity, one should begin
  as described in Section~\ref{sec:max_flow_prelim} by crudely approximating the minimum
  cut, and then applying the above algorithm in a binary search.
Crudely analyzed, this incurs a multiplicative overhead of $O (\log m/\epsilon)$.

\begin{proof}[Proof of Lemma~\ref{thm:res_conv_bound}]
Suppose as before that $\Reff{\rr^i} \leq (1 - 7\epsilon) \frac {\mu^i}{\maxflowguess^2}$ 
  for all $1 \leq i \leq N$. 
By Lemma~\ref{lem:total_weight} and the fact that $\mu^0=m$, the total weight at the end of the 
  algorithm is bounded by
\begin{equation}\label{eq:mu_bound}
\mu^N \leq \mu^0 \exp\left(\frac {\epsilon(1 - 2\epsilon)}{\rhopc} N\right) = m \exp\left(\frac {\epsilon(1 - 2\epsilon)}{\rhopc} N\right) . 
\end{equation}

Let $A\subseteq \{1,\dots,N \}$ be the set of $i$ for which $\cong{\tf^i}{e} \leq \rhopc$ for all $e$, and let  $B\subseteq \{1,\dots,N \}$ be the set of $i$ for which there exists an edge $e$ with $\cong{\tf^i}{e} > \rhopc$.  Let $a=|A|$ and $b=|B|$, and note that $a+b = N$.  We will obtain a contradiction by proving bounds on $a$ and $b$ that add up to something less than $N$.

We begin by bounding $a$.  
By applying Lemma~\ref{lem:nu_increase} to all of the steps in $A$ and noting that $\nu^i$ 
  never decreases during the other steps, we  obtain
\begin{equation}\label{eq:nu_bound}
\nu^{N} \geq \exp\left( \frac {\epsilon(1 - \epsilon)}{\rhopc}a\right) \nu^0 
=\exp\left( \frac {\epsilon(1 - \epsilon)}{\rhopc}a\right)
\end{equation}
Since $\nu^N \leq \mu^N$, 
we can combine Equations~\eqref{eq:mu_bound} and~\eqref{eq:nu_bound} to obtain
$$
\exp\left( \frac {\epsilon(1 - \epsilon)}{\rhopc}a\right)
\leq 
 m \exp\left(\frac {\epsilon(1 - 2\epsilon)}{\rhopc} N\right). 
$$
Taking logs of both sides and solving for $a$ yields
$$ \frac {\epsilon(1 - \epsilon)}{\rhopc}a
 \leq
\frac {\epsilon(1 - 2\epsilon)}{\rhopc} N+\ln m ,
$$
from which we obtain
\begin{equation}\label{eq:a_bound}
 a
 \leq
\frac {1 - 2\epsilon}{1 - \epsilon} N+\frac {\rhopc}{\epsilon(1 - \epsilon)}\ln m 
\leq 
(1 - \epsilon) N+\frac {\rhopc}{\epsilon(1 - \epsilon)}\ln m
< (1 - \epsilon) N+\frac {7 \rhopc}{6 \epsilon}\ln m.
\end{equation}

We now use a similar argument to bound $b$.  
By applying Lemma~\ref{lem:eff_res_increase} to all of the steps in $B$ and noting that
  $\Reff{\rr^i}$ never decreases during the other steps, we  obtain
\[
\Reff{\rr^N} 
 \geq {\exp\left( \frac {\epsilon^2 \rhopc^2}{4 m}b \right) \Reff{\rr^0}}
 \geq  {\exp\left( \frac {\epsilon^2 \rhopc^2}{4 m}b \right) \frac{1+\epsilon}{m^2 {\maxflow}^2}}
 \geq  {\exp\left( \frac {\epsilon^2 \rhopc^2}{4 m}b \right) \frac{1+\epsilon}{m^2 {\maxflowguess}^2}},
\]
where the second inequality applies the bound $\Reff{\rr^0} \geq 1/(m^2 {\maxflow}^2)$, which follows from an argument like the one used in Equation~\eqref{eq:eff_res_lb}.
Using our assumption that $\Reff{\rr^N} \leq (1 - 7\epsilon) \frac {\mu^N}{\left(\maxflowguess\right)^2}$ 
  and the bound on $\mu^N$ from  Equation~\eqref{eq:mu_bound}, this gives us
$$\exp\left( \frac {\epsilon^2 \rhopc^2}{4 m}b \right) \frac{1+\epsilon}{m^2 {\maxflowguess}^2} \leq \frac{1-7\epsilon}{{\maxflowguess}^2} m \exp\left(\frac {\epsilon(1 - 2\epsilon)}{\rhopc} N\right). $$
Taking logs and rearranging terms gives us
\begin{equation}\label{eq:b_bound}
b
\leq \frac{4m(1-2\epsilon)}{\epsilon \rhopc^3}N+\frac{4m}{\epsilon^2 \rhopc^2} \ln\left(\frac{1-7\epsilon}{1+\epsilon} m^3\right)
<  \frac{4m}{\epsilon \rhopc^3}N+\frac{12m}{\epsilon^2 \rhopc^2} \ln m.
\end{equation}
Adding the inequalities in Equations~\eqref{eq:a_bound} and~\eqref{eq:b_bound}, grouping terms, and plugging in the values of $\rhopc$ and $N$,  yields 
\begin{align*}
N = a+b & <          
 \left((1 - \epsilon)+ \frac{4m}{\epsilon \rhopc^3}\right)N+
 \left(\frac {7 \rhopc}{6 \epsilon}+\frac{12m}{\epsilon^2 \rhopc^2}\right) \ln m\\
&=
\left((1 - \epsilon)+ \frac{4m}{\epsilon (27m \epsilon^{-2})}\right)(5\epsilon^{-8/3} m^{1/3} \ln m)+
  \left(\frac {7 m^{1/3} \epsilon^{-2/3}}{2 \epsilon}+
      \frac{12m}{\epsilon^2 (  9m^{2/3} \epsilon^{-4/3})  }\right) \ln m\\
&=\left((1 - \epsilon)+ \frac{4\epsilon}{27}\right)(5\epsilon^{-8/3} m^{1/3} \ln m)+
  \left(\frac {7 m^{1/3} }{2 \epsilon^{5/3}}+\frac{12m^{1/3}}{9\epsilon^{2/3}  }\right)\ln m\\
&=  \frac{5m^{1/3}}{\epsilon^{8/3}}\ln m -\left( \frac{41}{54}-\frac{12\epsilon}{9  }\right)  \frac{m^{1/3}}{\epsilon^{5/3}}\ln m \\
&<  \frac{5m^{1/3}}{\epsilon^{8/3}} = N. 
\end{align*} 
This is our desired contradiction, which completes the proof.
\end{proof}


\bibliographystyle{abbrv}
\bibliography{maxFlowPaper}
\appendix

\section{Computing Electrical Flows}\label{sec:linsolve}
\begin{proof}[Proof of Theorem~\ref{thm:computing_electrical_flow}]
Without loss of generality, we may scale the resistances so that
  they lie between $1$ and $R$.
This gives the following relation between $F$ and
  the energy of the electrical $s$-$t$ flow of value $F$:
\begin{equation}\label{eqn:energyF}
\frac{F^{2}}{m} \leq   \energy{\rr}{\ff} \leq F^{2} R m.
\end{equation}

Let $\LL$ be the Laplacian matrix for this electrical flow problem.
Note that all off-diagonal entries of $\LL$ lie between 
 $-1$ and  $-1/R$.
Recall that the vector of optimal vertex potentials $\phiphi$ of the 
  electrical flow is the solution to the following linear system:
\[
  \LL \phiphi = F \chist 
\]

For any $\epsilon > 0$, the algorithm of 
  Koutis, Miller and Peng~\cite{KoutisMP10}
  provides in time 
  $O\left(m \log^{2} m \log^{2} \log m \log \epsilon^{-1} \right)$ 
  a set of potentials $\hphiphi$ such that
\[
  \norm{\hphiphi  - \phiphi}_{\LL} \leq \epsilon \norm{\phiphi}_{\LL},
\]
where 
\[
  \norm{\phiphi}_{\LL} \defeq \sqrt{\phiphi^{T} \LL \phiphi}.
\]

Note also that
\[
  \norm{\phiphi}_{\LL}^{2} = \energy{\rr}{\ff},
\]
where
  $\ff$ is the potential flow corresponding
  to $\phiphi$.
Now, let $\hff$ be the potential flow corresponding to
  $\hphiphi$,
\[
  \hff = \CC \BB^{T} \hphiphi .
\]
The flow $\hff$ is not necessarily an $s$-$t$ flow because 
  $\hphiphi$ only approximately satisfies the linear system.
A small amount of flow can enter or leave every other vertex.
In other words, $\hff$ may not satisfy $\BB \hff  = F \bchist$.
Below, we will compute an approximation $\tff$ of $\hff$ that 
  satisfies $\BB \tff  = F \bchist$.

Before fixing this problem, we first observe that
  the energy of $\hff$ is not much more than the energy of $\ff$.
We have
\[
  \norm{\hphiphi}_{\LL}^{2} = \energy{\rr}{\hff},
\]
and by the triangle inequality
\[
  \norm{\hphiphi}_{\LL} \leq \norm{\phiphi}_{\LL} 
  +  \norm{\hphiphi  - \phiphi}_{\LL}
 \leq (1 + \epsilon )  \norm{\phiphi}_{\LL} .
\]
So
\[
  \energy{\rr}{\hff} \leq (1+ \epsilon)^{2} \energy{\rr}{\ff}.
\]

To see that the flow $\hff$ does not flow too much into or out of any
  vertex other than $s$ and $t$, note that the flows into and out of
  vertices are given by the vector
\[
  \ii_{ext} \defeq \BB \hff .
\]
Note that the sum of the entries in $\ii_{ext}$ is zero.
We will produce an $s$-$t$ flow of value $F$ by adding a flow
  to $\hff$ to obtain a flow $\tff$ for which
\[
   \BB \tff  = F \bchist .
\]

Let $\eta$ be the maximum discrepancy between $\ii_{ext}$ and $F \bchist$:
\[
  \eta \defeq \norm{\ii_{ext} - F \bchist}_{\infty}.
\]
We have
\[
  \eta \leq 
  \norm{\ii_{ext} - F \bchist}_2 = 
  \norm{\LL \hphiphi - \LL \phiphi}_2
 \leq 
  \norm{\LL}_2   \norm{\hphiphi - \phiphi}_{\LL}
  \leq 
   2 n \norm{\hphiphi - \phiphi}_{\LL}
  \leq 
   2 n \epsilon \sqrt{\energy{\rr}{\ff }}.
\]
It is an easy matter to produce an $s$-$t$ flow $\tff$
  that differs from $\hff$ by at most $n \eta$ on each edge.
For example, one can do this by solving a flow problem in
  a spanning tree of the graph.
Let $T$ be any spanning tree of the graph $G$.
We now construct a flow in $T$ that with the demands
  $F \bchist (u) -\ii_{ext} (u)$
  at each vertex $u$.
As the sum of the positive demands in this flow is at most $n \eta$, one can find
  such a flow in which every edge flows at most $n \eta$.
Moreover, since the edges of $T$ are a tree, one can find the
  flow in linear time.
We obtain the flow $\tff$ by adding this flow to $\hff$.
The resulting flow is an $s$-$t$ flow of value 
  $F$.
Moreover,
\[
  \norm{\hff - \tff}_{\infty} \leq n \eta.
\]
To ensure that we get a good approximation of the electrical flow,
  we  now impose the requirment that
\[
  \epsilon \leq  \frac{\delta}{2 n^{2} m^{1/2} R^{3/2}}.
\]
To check that requirement $a$ is satisfied, observe that
\begin{align*}
  \energy{\rr}{\tff} 
& = 
  \sum_{e} r_{e} \tf_{e}^{2}
\\
&  \leq 
  \sum_{e} r_{e} (\hff_{e} + n \eta )^{2}
\\
&  \leq 
  \energy{\rr}{\hff}
  +
  2 n \eta \sum_{e} r_{e} \hf_{e}
 +
   n^{2} \eta^{2} \sum_{e} r_{e}  
\\
& \leq 
  \energy{\rr}{\hff}
 +
 (2 n \eta F + n^{2} \eta^{2}) \sum_{e} r_{e}  
\\
& \leq 
  \energy{\rr}{\hff}
 +
 \epsilon (2 n^{2} (mR)^{1/2} + 4 n^{4}) \energy{\rr}{\ff} m R
&
\text{(by \eqref{eqn:energyF})}
\\
& \leq 
\energy{\rr}{\ff}
\left((1+\epsilon)^{2} + \epsilon 6 n^{4} m R^{3/2} \right).
\end{align*}
We may ensure that this is at most $(1+\delta) \energy{\rr}{\ff}$
  by setting
\[
  \epsilon = \frac{\delta}{12 n^{4} m R^{3/2}}.
\]

To establish part $b$,
  note that
\[
  \norm{\phiphi}_{\LL} = 
  \sum_{e} r_{e} f_{e}^{2}
\]
and so
\[
  \sum_{e} r_{e} (f_{e} - \hf_{e})^{2}
= 
  \norm{\phiphi - \hphiphi }_{\LL}^{2}
\leq 
  \epsilon^{2} \energy{\rr}{\ff}.
\]
We now bound $\abs{r_{e} f_{e}^{2} - r_{e} \hf_{e}^{2}}$
  by
\begin{align*}
\abs{r_{e} f_{e}^{2} - r_{e} \hf_{e}^{2}}^{2}
& =
r_{e} \left(f_{e} - \hf_{e} \right)^{2} r_{e} \left(f_{e} + \hf_{e} \right)^{2}
\\
& \leq 
\epsilon^{2} \energy{\rr}{\ff}  (2+\epsilon)^{2}  \energy{\rr}{\ff},
\end{align*}
which implies
\[
\abs{r_{e} f_{e}^{2} - r_{e} \hf_{e}^{2}} 
\leq 
\epsilon  (2+\epsilon)  \energy{\rr}{\ff}.
\]
It remains to bound 
$\abs{r_{e} \hf_{e}^{2} - r_{e} \tf_{e}^{2}} $,
which we do by the calculation
\begin{align*}
\abs{r_{e} \hf_{e}^{2} - r_{e} \tf_{e}^{2}} 
& \leq 
\sqrt{R}
\abs{\hf_{e} - \tf_{e}}
\sqrt{r_{e}}
\abs{\hf_{e} + \tf_{e}}
\\
& \leq 
\sqrt{R}
2 n^{2} \epsilon \sqrt{\energy{\rr}{\ff}}
(2 + \epsilon) \sqrt{\energy{\rr}{\ff}}
\\
& = 
\epsilon \sqrt{R}
2 (2 + \epsilon) n^{2} 
 \energy{\rr}{\ff}.
\end{align*}
Putting these inequalities together we find
\[
\abs{r_{e} f_{e}^{2} - r_{e} \tf_{e}^{2}} 
\leq 
\epsilon 
\sqrt{R}
(2n^{2} + 1) (2 + \epsilon) 
 \energy{\rr}{\ff}
\leq 
\energy{\rr}{\ff} 
\frac{\delta}{2 m R}.
\]

To prove part $c$, first recall that 
\[
  \LL \phiphi = F \chist 
\]
and that
\[
  \energy{\rr}{\ff} = F^{2} \Reff{\rr}.
\]
So,
\begin{align*}
\epsilon^{2}   \norm{\phiphi}_{\LL}^{2} 
& \geq 
  \norm{\phiphi - \tphiphi}_{\LL}^{2}
\\
& = 
   \left(\phiphi - \tphiphi \right)^{T} \LL \left(\phiphi - \tphiphi \right)
\\
& = 
  \norm{\tphiphi}_{\LL}^{2}
+  \norm{\phiphi}_{\LL}^{2}
 - 2  
   \tphiphi^{T} \LL \phiphi
\\
& = 
  \norm{\phiphi}_{\LL}^{2}
+  \norm{\tphiphi}_{\LL}^{2}
 - 2  
  F \chist^{T} \tphiphi .
\end{align*}
Thus,
\begin{align*}
 2  
  F \chist^{T} \tphiphi
& \geq 
  \norm{\phiphi}_{\LL}^{2}
+  \norm{\tphiphi}_{\LL}^{2}
 - \epsilon^{2}   \norm{\phiphi}_{\LL}^{2}
\\
&\geq 
  \left(1 + (1-\epsilon)^{2} - \epsilon^{2} \right)
   \norm{\phiphi}_{\LL}^{2}
\\
&=
  \left(2 - 2 \epsilon \right)
   \norm{\phiphi}_{\LL}^{2}
\\
&=
  \left(2 - 2 \epsilon \right)
  F^{2} \Reff{\rr}.
\end{align*}
\end{proof}


\end{document}